\documentclass[sigplan,screen]{acmart}
\acmYear{2022}\copyrightyear{2022}
\setcopyright{rightsretained}
\acmConference[PPoPP '22]{27th ACM SIGPLAN Symposium on Principles and Practice of Parallel Programming}{April 2--6, 2022}{Seoul, Republic of Korea}
\acmBooktitle{27th ACM SIGPLAN Symposium on Principles and Practice of Parallel Programming (PPoPP '22), April 2--6, 2022, Seoul, Republic of Korea}
\acmPrice{}
\acmDOI{10.1145/3503221.3508412}
\acmISBN{978-1-4503-9204-4/22/04}

\usepackage{booktabs}   
\usepackage{subcaption} 
\usepackage{comment}

\usepackage{dirtree}

\usepackage[linesnumbered,ruled,vlined]{algorithm2e}
\SetArgSty{textnormal}
\SetAlFnt{\small}
\SetAlCapFnt{\small}
\SetAlCapNameFnt{\small}
\DontPrintSemicolon
\usepackage{listings}
\usepackage{subcaption}
\usepackage{paralist}
\usepackage[normalem]{ulem}

\usepackage[]{todonotes}
\excludecomment{cond}

\newcommand{\toremove}[1]{}
\begin{cond}
\renewcommand{\toremove}[1]{{\color{red}#1}}
\end{cond}

\newcommand{\cradd}[1]{#1}
\newcommand{\crremove}[1]{}

\begin{document}

\title{Bundling Linked Data Structures for Linearizable Range Queries}         


\author{Jacob Nelson-Slivon}
\affiliation{
  \institution{Lehigh University, USA}            
  \country{}
}
\email{jjn217@lehigh.edu}          

\author{Ahmed Hassan}
\affiliation{
  \institution{Lehigh University, USA}           
  \country{}
}
\email{ahh319@lehigh.edu}         

\author{Roberto Palmieri}
\affiliation{
  \institution{Lehigh University, USA}           
  \country{}
}
\email{palmieri@lehigh.edu}

\begin{abstract}
We present bundled references, a new building block to provide linearizable range query operations for highly concurrent lock-based linked data structures. Bundled references allow range queries to traverse a path through the data structure that is consistent with the target atomic snapshot. We demonstrate our technique with three data structures: a linked list, skip list, and a binary search tree. Our evaluation reveals that in mixed workloads, our design can improve upon the state-of-the-art techniques by 1.2x-1.8x for a skip list and 1.3x-3.7x for a binary search tree. We also integrate our bundled data structure into the DBx1000 in-memory database, yielding up to 40\% gain over the same competitors.
\end{abstract}

\ccsdesc[500]{Computing methodologies~Concurrent algorithms}
\ccsdesc[500]{Information systems~Data structures}

\keywords{Concurrent Data Structures, Range Queries, Fine-grain Synchronization}  

\maketitle

\section{Introduction}

Iterating over a collection of elements to return those that
fall within a contiguous range (also known as a \textit{range query} operation) is an essential feature for data repositories.
In addition to database management systems, which historically deploy support for range queries (through predicate reads or writes), recent key-value stores (e.g., RocksDB~\cite{facebook2020rocksdb} and others~\cite{DBLP:conf/sigcomm/EscrivaWS12,DBLP:conf/eurosys/MaoKM12,silk,pebblesdb,google2019leveldb}) enrich their traditional APIs to include range query operations.

With the high-core-count era in full swing, providing high-performance range query operations that execute concurrently with modifications is challenging.
On the one hand, ensuring strong correctness guarantees of range queries, such as linearizablity~\cite{linearizability}, requires that they observe a consistent snapshot of the collection regardless of any concurrent update that may take place. On the other hand, since range queries are naturally read-only operations, burdening them with synchronization steps to achieve strong correctness guarantees may significantly deteriorate their performance.


In this paper we propose \textit{bundled references}\footnote{This work builds on an initial design appeared in~\cite{bundling-poster-ppopp21}.}, a new building block to design lock-based linearizable concurrent linked data structures optimized to scale up performance of range query operations executing concurrently with update operations. 
The core innovation behind bundled reference lies in adapting the design principle of 
Multi Version Concurrency Controls (MVCC)~\cite{wu2017empirical, MV-STM} to lock-based linked data structures, and improving it by eliminating the overhead of determining the appropriate version to be returned.
Bundled references achieve that by augmenting each link in a data structure with a record of its previous values, each of which is tagged with a timestamp reflecting the point in (logical) time when the operation that generated that link occurred. In other words, we associate timestamps to references connecting data structure elements instead of the nodes themselves.

The bundled reference building block enables the following characteristics of the data structure:
\begin{itemize}
\item Range query operations \crremove{are linearized} \cradd{establish their snapshot} just before reaching the first element in the range, which helps reduce interference \cradd{from} \crremove{with} ongoing and subsequent update operations \cradd{on the range};

\item Data structure traversals, including those of contains and update operations, may proceed uninstrumented
until reaching the desired element(s).

\item Well-known memory reclamation techniques, such as EBR~\cite{ebr}, can be easily integrated into the bundled references to reclaim data structure elements, which reduces the space overhead of bundling.
\end{itemize}


We demonstrate the efficacy of bundling by applying it to three widely-used lock-based ordered Set implementations, namely a linked list, a skip list, and a binary search tree (BST).
While the linked list is a convenient data structure to illustrate the details of our design that favors range query operations, the skip list and the BST are high-performance data structures widely used in systems (such as database indexes) where predicate reads are predominant.
In these new data structure we augment the existing links with bundled references to provide linearizable range queries.

In a nutshell, the history of a link between nodes (called a \textit{bundle}) is consistently updated every time a successful modification to the data structure occurs, and its entries are labeled according to a global timestamp.
Data structure traversals, including those of range query and contains operations, do not need to use bundles until they reach their target range or element. Then, range queries read the global timestamp and follow a path made of the latest links marked with a timestamp lower than (or equal to) the previously read timestamp. Contains operations also use bundles after the non-instrumented traversal to consistently read the target element, but they do not need to access the global timestamp.



Bundling follows the trend of providing abstractions and techniques to support range query operations on highly concurrent data structures~\cite{java-util-concurrent-lib,rlu,ebr-rq,snapcollector, DBLP:conf/icpp/RodriguezS20}. 
Among them, the most relevant competitors include read-log-update (RLU)\cite{rlu}, the versioned CAS (vCAS) approach~\cite{vcas-ppopp21}, and a solution based on epoch-based reclamation (EBR-RQ)~\cite{ebr-rq}. We contrast their functionalities in the related work section and their performance in the evaluation section.


\toremove{Analyzing the performance results, we found that in a mixed workload, bundling offers up to 1.4x and 3.7x improvement over the closest and next closest competitors.
\textbf{Further, bundling achieves a more consistent performance profile across different configurations than RLU and EBR-RQ, whose design choices can lead them to prefer specific workloads.
Compared with vCAS, bundling's performance prevails in read-dominated workload, while in write-dominated workloads, vCAS equals (or surpasses) bundling.}} 

In summary, analyzing the performance results we found that bundling achieves a more consistent performance profile across different configurations than RLU and EBR-RQ, whose design choices can lead them to prefer specific workloads. Concerning mixed workloads, bundling offers up to 1.8x improvement over EBR-RQ, and up to 3.7x improvement over RLU.
Compared with vCAS, bundling's performance prevails in all linked list experiments (reaching up to 2x improvement). Also, bundling outperforms vCAS in both skip list and BST for read-dominant cases, reaching up to 1.5x improvement. In write-dominated workloads, vCAS equals bundling in skip list, and outperforms it in BST at high thread count.
We also integrate our bundled skip list and BST as indexes in the DBx1000 in-memory database and test them using the TPC-C benchmark. We find that bundling provides up to 1.4x better performance than the next best competitor at high thread count.

\section{Related Work}
\label{sec:rel-work}

\textbf{Linearizable range queries.}
Existing work has focused on providing range queries through highly-specific data structure implementations~\cite{karytree,ctrie,catree,leaplist,DBLP:conf/ppopp/BronsonCCO10,kiwi}.
While recognizing their effectiveness,
their tight dependency on the data structure characteristics makes them difficult to extend to other structures, even if manually. 
The literature is also rich with effective concurrent data structure designs that lack range query support and cannot leverage the above data structure specific solutions to perform range queries. 
This motivates generalized solutions, which achieve linearizable range queries by applying the same technique to various data structures~\cite{ebr-rq, rlu, snapcollector, bundling-poster-ppopp21}.



Read-log-update (RLU)~\cite{rlu} is a technique in which writing threads keep a local log of updated objects, along with the logical timestamp when the update takes effect. 
When no reader requires the original version, the log is committed.
It extends read-copy-update (RCU)~\cite{rcu} to support multiple object updates. 
Range queries using RLU are linearized at the beginning of their execution, after reading a global timestamp and fixing their view of the data structure.
However, in RLU, updates block while there are ongoing RLU protected operations,
as it only commits its changes after guaranteeing no operation will access the old version.
Bundling minimizes write overhead because new entries are added while deferring the removal of outdated ones.




Snapcollector~\cite{snapcollector} logs changes made to the data structure during an operation's lifetime so that concurrent updates are observed.
A range query first announces its intention to snapshot the data structure by posting a reference to an object responsible for collecting updates.
It traverses as it would in a sequential setting, then checks a report of concurrent changes it may have missed.
The primary difference with respect to RLU is that range queries are linearized at the end of the operation, after disabling further reports.

Although the construction of Snapcollector is wait-free, this method may lead a range query to observe reports of changes that were already witnessed during its traversal.
Creating and announcing reports penalizes operation performance; not to mention the memory overhead required to maintain these reports.
We experimentally verify that the cost of these characteristics is high.
With our bundling approach, a range query visits nodes in the range only once to produce its view of the data structure and is linearized when it reads the global timestamp right before entering the range.
An extension of Snapcollector enables snapshotting only a range of the data structure instead of all elements~\cite{chatterjee2017lock}.
However, this approach continues to suffer many of the same pitfalls as the original design. 
In addition to these, concurrent range queries with overlapping key ranges are disallowed.

Arbel-Raviv and Brown~\cite{ebr-rq} build upon epoch-based memory reclamation (EBR) to provide linearizable range queries.
In this method, range query traversals leverage a global timestamp to determine if nodes, annotated with a timestamp by update operations, belong in their snapshot.
In order to preserve linearizability, remove operations announce their intention to delete a node before physically removing them and adding them to the list of to-be-deleted nodes, or limbo list. 
Range queries scan the data structure, the announced deletions, and limbo list to determine which nodes to include in their view, potentially resulting in a situation where nodes are observed multiple times.
The design also prioritizes update-mostly workloads, since range queries' timestamp updates conflict.
Our bundling approach enhances performance of range queries by allowing them to traverse only the nodes in the range without needing to validate its snapshot.

Bundling, like the competitors included thus far, focuses solely on supporting linearizable range queries.
Other techniques target linearizable bulk update operations, which mutate several data structure elements atomically~\cite{DBLP:conf/icpp/RodriguezS20}.
These methodologies are not optimized specifically for range queries and therefore are not included in our comparisons.

\textbf{MVCC.}
Multi-version concurrency control (MVCC), widely used in database management systems, relies on timestamps to 
coordinate concurrent accesses.
Many different implementations exist~\cite{ben2019multiversion, bernstein1983multiversion, neumann2015fast, larson2011high, lim2017cicada}; all rely on a multiversioned data repository where each shared object stores a list of versions,
and each version is tagged with a creation timestamp.
Transactions then read the versions of objects that are consistent with their execution.





The de facto standard for version storage in MVCC systems is to maintain a list of versions (\textit{version list}) for each object that is probed during a read~\cite{wu2017empirical}, with innovations targeting this particular aspect.
One example, Cicada~\cite{lim2017cicada}, uses a similar idea of installing \texttt{PENDING} version for every written object as the first step in its validation phase. The main difference in bundling is that pending entries only exist for a short duration surrounding the linearization point. X-Engine~\cite{huang2019x} and KiWi~\cite{kiwi} use version lists for objects, rather than links like in bundling, which cause them to visit nodes not belonging to their snapshot.
Multi-versioned Software Transactional Memory~\cite{MV-STM} applies MVCC on in-memory transactions. However, it inherently requires \crremove{speculative accesses} \cradd{update transactions to instrument all accesses (and validate them)} while bundling exploits data structure specific semantics to avoid false conflicts and improve performance.

\textbf{Persistent data structures.}
The bundled reference abstraction is similar in spirit to the concept of \textit{fat nodes} in \textit{persistent data structures}~\cite{driscoll1986making}. 
In principle, persistent data structures are those which maintain all previous versions of the data structure.
Bundling aims at providing efficient linearizable range queries in highly concurrent workloads, while persistent data structures are commonly used in functional programming languages to maintain theoretical requirements regarding object immutablility~\cite{okasaki1999purely, hickey2008clojure} and algorithms requiring reference to previous state~\cite{sarnak1986planar, chien2001efficient}.


\cradd{
\textbf{vCAS.}
vCAS~\cite{vcas-ppopp21} implements range queries for lock-free data structures 
while retaining the non-blocking guarantee of data structures through their so called vCAS objects. Like bundling, each vCAS object maintains a versioned list of values but for CAS objects. Operations can query the state of this vCAS object at a certain snapshot using a \textit{versioned read} interface.
Since CAS objects in lock-free data structures are often the links between data structure nodes, vCAS enables traversals to observe links at a given snapshot, which can then be leveraged to implement range queries.

Applying the vCAS technique to lock-free data structures is straightforward since it only requires replacing CAS objects with vCAS objects. On the other hand, the requirements when applying the vCAS technique to lock-based data structures are unclear since there is no simple rule to determine which components should be replaced by vCAS objects. In principle, there is a mismatch between the target abstraction of vCAS and its application to lock-based data structures. In contrast, bundling is designed specifically to represent references and more precisely captures the meaning of versioning in the context of linked data structures. In our evaluation, we demonstrate vCAS can be ported to lock-based data structures by using vCAS objects in place of both pointer(s) and metadata (e.g., flags for logical deletion), but not locks. Despite that, it is still not clear whether vCAS is broadly applicable to lock-based data structures.

The replacement strategy required for vCAS illuminates another important difference.
Bundling augments references in the data structure and keeps the original links unchanged, while vCAS replaces them.
Because of this, bundling provides flexibility in how operations traverse the data structure, either using bundles or the original links, which is a fundamental enabler for our optimized traversals.
For vCAS, the result of tightly coupling the versioned CAS object to the underlying field is that all operations must pass through the vCAS API, which leads to additional costs for vCAS when traversals are long (see Section~\ref{sec:evaluation}).
Whether or not our optimization can be applied in the case of vCAS is an open question, but it would likely include non-trivial changes to the current abstraction.
Similarly, making bundling lock-free is not an insignificant challenge and potentially requires replacing locks with multi-CAS operations to atomically update each link along with its augmented bundle. Unifying the two is an intriguing future direction.
}

\cradd{
\textbf{Optimized Traversals.}
 Among other publications on optimized traversals~\cite{otblist,lazylist,lazyskiplist,citrus,afek2011towards,crain2013contention,david2015asynchronized}, two recent works are worth mentioning. Like their predecessors, both capture the notion of data structure traversals and their logical isolation from an operation's ``real'' work. In the first, Freidman et al.~\cite{friedman2020nvtraverse} define a Traversal Data Structure as a lock-free data structure satisfying specific properties that allow operations to be split into two phases. These phases are the Traversal phase and the Critical phase, of which only the latter requires instrumentation for making a data structure consistently persistent in non-volatile memory. This resembles our notions of pre-range and enter-range, but is applied specifically in the context of lock-free data structures and non-volatile memory. Similarly, Feldman et al.~\cite{feldman2020proving} construct a proof model that is centered around proving the correctness of optimistic traversals. The model combines single-step compatibility with the forepassed condition to demonstrate that, for a given operation, a node has been reachable at some point during the traversal. An interesting extension would be to understand the applicability of this model when range queries are added to the data structure APIs. Due to versioning, the current conditions may not fully capture the requirements to prove that range queries are linearizable.}

\section{The Bundle Building Block}
\label{sec:overview}




The principal idea behind \textit{bundling} is the maintenance of a historical record of physical changes to the data structure so that range queries can traverse a consistent snapshot.
As detailed below, the idiosyncrasy of bundling is that this historical record stores \textit{links} between data structure elements that are used by range query operations to rebuild the exact composition of a range at a given (logical) time.

Update operations are totally ordered using a global timestamp, named \texttt{globalTs},
which is incremented every time a modification to the data structure takes place \crremove{(i.e., when an update operation reaches its linearization point)}.

Every link in the data structure is backed by a \textit{bundle}, implemented as a list of \textit{bundle entries} (Listing~\ref{lst:classes1}). Each bundle entry logs the value of the link and the value of \texttt{globalTs} at the time the link was added to the bundle.
Whenever an update operation reaches its linearization point, meaning when it is guaranteed to complete, it prepends a bundle entry consisting of the newest value of the link and the value of the global timestamp.
Because of this, the head of the bundle always reflects the link's latest value.

It is worth noting that other solutions~\cite{vcas-ppopp21,ebr-rq} also use a global timestamp but assign the responsibility of incrementing it to range query operations. 
We experimentally verify (see our evaluation section) that incrementing the global timestamp upon completion of range query or update operation is irrelevant in terms of performance for bundling. Our decision of letting update operations increment \texttt{globalTs} simplifies traversals with bundles since it allows each bundle entry to be tagged with a unique timestamp.

Since each link's history is preserved through the bundles, range queries simply need to reach the range, read the global timestamp, and traverse range of the linked data structure using the newest values no larger than the observed global timestamp.
This design is inherently advantageous when pruning bundle entries. 
In fact, a bundle entry may be removed (or recycled)
if an entry is no longer the newest in the bundle and no range query needs it.

\begin{figure}[t]
    \centering
    \includegraphics[width=0.45\textwidth]{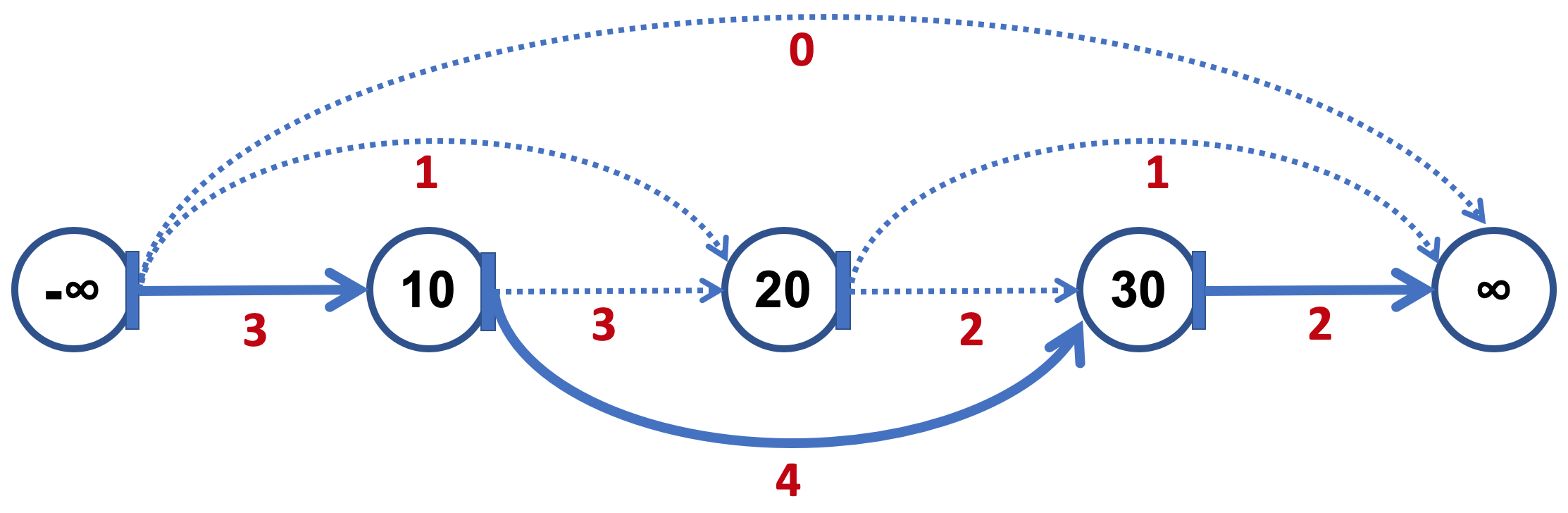}
    \caption{An example of using bundled references in a linked list. The path made of solid lines represents the state of the linked list after all update operations take place. Edges are labeled with their respective timestamps.}
    \label{fig:bundle-example}
\end{figure}

Figure~\ref{fig:bundle-example} shows an example on how bundles are deployed in a linked list.
As shown in the figure, the next pointer of each node is replaced by a bundle object that encapsulates the history of this next pointer. The figure shows the state of the linked list and its bundles after the following sequence of operations (on an empty linked list): \texttt{insert(20), insert(30), insert(10), remove(20)}.

To understand how this state is generated, we assume that the list is initialized with a single bundle reference whose timestamp is ``0'' (the initial value of \texttt{globalTs}), which connects its head and tail sentinel nodes. Inserting \texttt{20} does not replace this reference. Instead, it creates a new entry in the head's bundle with timestamp ``1'' pointing to the newly inserted node as well as an entry with the same timestamp in this new node pointing to the tail node. Similarly, inserting \texttt{30} and \texttt{10} adds new bundle entries with timestamps  ``2'' and ``3'', respectively. The last operation that removes \texttt{20} also does not replace any reference. Instead, it creates a new bundle entry in \texttt{10}'s bundle (with timestamp ``4'') that points to \texttt{30}, which reflects physically deleting \texttt{20} by updating the predecessor's next pointer.

Now assume that different range queries start at different times concurrently with those update operations. For clarity, we name a range query $R_i$
if it reads the value $i$ of \texttt{globalTs}, and for simplicity we assume its range matches the entire key range. Regardless of the times at which the different nodes are traversed, each range query is always able to traverse the proper snapshot of the list that reflects the time it started. For example, $R_0$ will skip any links in the range added after it started because all of them have timestamp greater than ``0''. Also, $R_3$ will observe \texttt{20} even if it reaches \texttt{10} after \texttt{20} is deleted. This is because in that case it will use the bundle entry whose timestamp is ``3'', which points to \texttt{20}.

The solid lines in the figure represent the most recent state of the linked list. Different insights can be inferred from this solid path. First, the references in this path are those with the largest timestamp in each bundle. This guarantees that any operation (including range queries) that starts after this steady state observes the most recent abstract state of the list. Second, once the reference with timestamp ``4'' is created, \texttt{20} becomes no longer reachable by any operation that will start later, because this operation will observe a timestamp greater than (or equal to) ``4''. Thus,
unreachable elements can be concurrently reclaimed.

\vspace{-2pt}
\subsection{Bundle Structure}
\label{sec:bundledref}

\textit{Bundling} a data structure entails augmenting its links with a bundle to produce a bundled reference.
Importantly, bundling does not replace links. For example, in Listing~\ref{lst:lazylistnode}, each node in a linked list has a \texttt{next} field and an associated \texttt{bundle} field.
A design invariant that we maintain to support this behavior is that the value of a link is equal to the newest entry in its corresponding bundle. 
Later, we describe how duplicating the newest link can be leveraged to improve performance when bundles are not needed to be traversed.

As shown in Listing~\ref{lst:classes1}, a \textit{bundle} is a collection of entries.
Each bundle entry must contain a pointer value, \texttt{ptr}, and the timestamp associated with this value, \texttt{ts}.

Bundles are accessed using a \texttt{DereferenceBundle} API, which returns the value of the bundled reference that \textit{satisfies} a given timestamp.
We say an entry \textit{satisfies} timestamp $ts$ if it is the newest entry in the bundle when the global timestamp equaled $ts$, which is the bundle entry with the largest timestamp that is less than or equal to $ts$.

\begin{minipage}{.22\textwidth}
\begin{lstlisting}[language=C++, frame=single, escapechar=|, basicstyle=\scriptsize,  numbers=left, stepnumber=1, numbersep=5pt, xleftmargin=5pt, framexleftmargin=7pt, label=lst:classes1, caption= Bundle.]
timestamp_t globalTs;
class BundleEntry {|\label{line:bundleentry1}|
    Node * ptr;
    timestamp_t ts;
    BundleEntry * next;
}|\label{line:bundleentry2}|
class Bundle {|\label{line:bundle1}|
    BundleEntry * head;
}|\label{line:bundle2}|
\end{lstlisting}
\end{minipage}
\hspace{5pt}
\begin{minipage}{.22\textwidth}
\begin{lstlisting}[language=C++, frame=single, escapechar=|, basicstyle=\scriptsize,  numbers=left, stepnumber=1, numbersep=5pt, xleftmargin=5pt, framexleftmargin=7pt, label=lst:lazylistnode, caption= Linked List Node.]
class Node {|\label{line:node1}|
    key_t key;
    val_t val;
    lock_t lock;
    bool deleted;
    // Bundled reference.
    Node * next;
    Bundle bundle;
}|\label{line:node2}|

\end{lstlisting}
\end{minipage}

In order to perform a traversal, each bundled data structure must implement two functions to determine the next node in the structure.
The first is \texttt{GetNext}, which uses the original links; the second is \texttt{GetNextFromBundle}, which uses \texttt{DereferenceBundle} internally. 
More details about these two functions will follow.

We implement our bundles as a
list of entries, sorted by their timestamp, but note that its implementation is orthogonal to the bundling framework as long as it supports the \texttt{DereferenceBundle} API.
So far we have assumed that a bundle may hold infinite entries. We address memory reclamation in Section~\ref{sec:memreclamation}.


\subsection{Bundles and Update Operations}
\label{sec:updates}

Generally speaking, an update operation has two phases.
The operation first traverses the data structure to reach the desired location where the operation should take place, then performs the necessary changes.
In bundling, only places where pointers are changed must be supplemented with operations on bundles; the operation's traversal procedure remains unaltered.

\begin{algorithm}[h]
\begin{scriptsize}
\KwIn{The set of bundles to update, $bundles$; the bundles' new values, $ptrs$}
\Begin{
\For{($b$,$p$) \textbf{in} ($bundles$, $ptrs$)}{
    $newEntry \leftarrow$ \textbf{new} $BundleEntry(ptr)$  \tcp*[f]{Pending entry} \\
    \label{line:pending}
    \While{true}{
        $newEntry.next \leftarrow expected \leftarrow bundle.head$ \\
        \lWhile{$bundle.head.ts = $ PENDING\_TS}{$\{\}$}
        \label{line:updatewaitpendingbundle}
        \lIf{AtomicCompareAndSwap(\&$bundle.head$, $expected$, $newEntry$)}{
        \label{line:preparecas}
            \textbf{break} \\
        }
    }
}
\Return AtomicFetchAndAdd(\&$globalTs$, 1) + 1
\label{line:fetchaddts}
}
\caption{PrepareBundles}
\label{algo:prepare}
\end{scriptsize}
\end{algorithm}
\vspace{-28pt}
\begin{algorithm}[h]
\begin{scriptsize}
\KwIn{A bundle to prepare, $bundle$; the linearization timestamp, $ts$}
\Begin{
    \For{$b$ \textbf{in} $bundles$}{
        $b.ts \gets ts$ \\
        \label{line:finalize}
    }
}
\caption{FinalizeBundles}
\label{algo:finalize}
\end{scriptsize}
\end{algorithm}

The goal for updates is to reflect changes, observable at the operation's completion, so that range queries can see a consistent view of the data structure. 
This is performed by book-ending an update's original critical section, which is defined at some point after all necessary locks are held, with \texttt{PrepareBundles} and \texttt{FinalizeBundles}, resulting in a three step process:
\begin{enumerate}
    \item \texttt{PrepareBundles} (Algorithm~\ref{algo:prepare}) inserts in the bundles of updated nodes a new entry in a \textit{pending} state (Line~\ref{line:preparecas}).
    Then, \texttt{globalTs} is atomically fetched and incremented and its new value is returned (Line~\ref{line:fetchaddts}).
    \item Next, the operation's original critical section is executed, making the update visible to other operations.
    \item Lastly, \texttt{FinalizeBundles} (Algorithm~\ref{algo:finalize}) sets the linearization timestamp of the pending entries by annotating them with the newly incremented timestamp, before releasing locks.
\end{enumerate}

The \texttt{PrepareBundles} step is crucial in the correctness of bundling. As detailed in Section~\ref{sec:rangequeries}, when a read operation traverses a node using a bundle, it waits until any pending entry is finalized to guarantee that it does not miss a concurrent update that should be included in its linearizable snapshot.
Additionally, concurrent updates attempting to add their own pending entry must block until the ongoing update is finalized (Algorithm~\ref{algo:prepare}, Line~\ref{line:updatewaitpendingbundle}). 
This is done so that concurrent updates to the same bundle are properly ordered by timestamp.
It is also possible to address this problem by assuming that all nodes whose bundles will change are locked.
We choose not to do so to make our design independent of data structure specific optimizations (see Section~\ref{sec:lazylist}).
Also note that incrementing \texttt{globalTs} must occur before any newly inserted nodes are reachable otherwise the linearizability of range queries may be compromised.


\subsection{Bundles and Range Queries}
\label{sec:rangequeries}


We consider range queries to have the following three phases:
\begin{enumerate}
    \item \texttt{pre-range}: the traversal of data structure nodes until the start of the range. 
    \item \texttt{enter-range}: the traversal from a node just outside the range to the first node falling within the range.
    \item \texttt{collect-range}: the traversal of the range, during which the result set is constructed.
\end{enumerate} 

Algorithm~\ref{algo:rangequery} outlines a range query for a bundled data structure, reflecting these steps.

\underline{Pre-range.}
The \texttt{pre-range} traversal is performed to locate the node immediately preceding the range (Lines~\ref{line:pre-range-start}-\ref{line:pre-range-end}).
Note that \texttt{GetNext} is a data structure specific procedure that returns the next node in a traversal toward the range \emph{without using bundles}.
For example, in the case of a binary search tree this would be the appropriate child link, not its corresponding bundle.
Using an uninstrumented traversal reduces overhead and improves performance by avoiding the costs associated with accessing bundles, as well as blocking on pending entries.
Once the predecessor to the range is located (Line~\ref{line:pre-range-break}), the \texttt{enter-range} phase is executed.

\underline{Enter-range.}
In sequential algorithms, the \texttt{enter-range} phase is trivial, as it is simply a dereference. 
However, for bundling it represents a critical moment that must be treated carefully since it is here that we switch to traversing a linearizable snapshot through bundles.


\begin{algorithm}
\begin{scriptsize}
\KwIn{Inclusive range to query, [$low$, $high$]; sentinel node and entry point to the data structure, $root$}
\KwOut{Set of nodes in the given range}
\Begin{
    \While{$true$}{
    \label{line:while-true}
        $pred \leftarrow root$ \tcp*[f]{Pre-range} \\
        $curr \leftarrow pred$.GetNext($low$, $high$) \\
        \label{line:pre-range-start}
        \label{line:enter-snapshot-nonrestarting}
        \While{$curr \ne nullptr$}{
            \label{line:pre-range}
            $pred \leftarrow curr$ \\
            $curr \leftarrow$ $curr$.GetNext($low$, $high$) \\
            \label{line:pre-range-next}
            \label{line:pre-range-next-nonrestarting}
            \lIf{$curr.key$ is in the range}{\textbf{break}}\label{line:pre-range-break}
        }
        \label{line:pre-range-end}
        
        $ts \leftarrow globalTs$ \tcp*[f]{Enter-range}\\
        \label{line:readglobalts}
        $curr$ $\leftarrow pred$.GetNextFromBundle($low$, $high$, $ts$) \\
        \label{line:enter-range}
        \While{$curr$ is not in the range}{
            $curr$ $\leftarrow curr$.GetNextFromBundle($low$, $high$, $ts$)
            \label{line:enter-range-while}
        }
        \label{line:collect-range-start}
        \Return CollectRange($curr$, $low$, $high$, $ts$) \tcp*[f]{Collect-range}
        \label{line:collect-range-collect}
    }
}
\caption{RangeQuery}
\label{algo:rangequery}
\end{scriptsize}
\end{algorithm}
\vspace{-25pt}
\begin{algorithm}[h]
\begin{scriptsize}
\cradd{
\KwIn{A bundle to dereference, $bundle$; the linearization timestamp, $ts$}
\Begin{
    $entry \gets bundle.head$ \\
    \lWhile{$entry.ts = $ PENDING\_TS}{$\{\}$\tcp*[f]{Maybe wait on update}}
    \lWhile{$entry.ts$ > $ts$} {
        $entry \gets entry.next$
    }
    \Return $entry.ptr$
}
}
\caption{DereferenceBundle}
\label{algo:dereferencebundle}
\end{scriptsize}
\end{algorithm}

To do so, a range query first fixes its linearizable snapshot of the data structure by reading \texttt{globalTs} into a local variable $ts$ (Line~\ref{line:readglobalts}).
This ensures that all future updates that increment the timestamp after the range query executes Line~\ref{line:readglobalts} are ignored.
Next, in Lines~\ref{line:enter-range}-\ref{line:enter-range-while}, the \texttt{enter-range} phase executes a loop of \texttt{GetNextFromBundle} to reach the first node in the range. This loop is needed to handle the cases in which nodes are inserted/deleted after the \texttt{pre-range} phase ends and before \texttt{globalTs} is read.



Internally, \texttt{GetNextFromBundle} uses \texttt{DereferenceBundle} \cradd{(Algorithm~\ref{algo:dereferencebundle})} to follow links.
Given a bundle and a timestamp $ts$, the function \texttt{DereferenceBundle} works as follows. 
The call first reads the head of the bundle and waits until the entry is not pending.
Next, it scans the bundle for the entry that satisfies $ts$ (i.e., the first entry whose timestamp is less than or equal to $ts$), returning the address of the node the entry points to.
Since it is written only by a single in-progress update, pending timestamps provide a low-contention linearization mechanism.

As we mentioned before, blocking while the most recent entry is pending is a necessary step to ensure that both range query and contains operations are correctly linearized along with update operations.
Because of this step, read operations observe updates only after the changes have become visible and the bundles are finalized.

It is worth noting that the transition from the \texttt{pre-range} phase to the \texttt{enter-range} phase is always consistent.
After the \texttt{pre-range} phase, a range query holds a reference to the predecessor of its range that was traversed to using normal links.
This node either has a pending bundle entry or a finalized one, since updates only make a node reachable after preparing the bundles and incrementing the global timestamp.
Because a range query reads \texttt{globalTs} only after reaching the node, it is guaranteed to find an entry satisfying its timestamp by calling \texttt{GetNextFromBundle}, even if it must wait for the concurrent update to finalize it first.

\underline{Collect-range.}
After the range is reached in the \texttt{enter-range} phase, the \texttt{collect-range} phase is performed by calling the data structure specific function \texttt{CollectRange}.
This function explores the data structure using bundles and the timestamp from Line~\ref{line:readglobalts} to generate the result set, returning the empty set if no nodes exist in the range.
As was the case for \texttt{GetNextFromBundle}, \texttt{CollectRange} uses \texttt{DereferenceBundle} to traverse the data structure and collect the result set.

A critical invariant that guarantees the correctness of bundling is that the \texttt{collect-range} phase always finds its required path through the bundles.
This is trivial if both deleted nodes and outdated bundle entries are never reclaimed because the entire history of each reference is recorded.
Section~\ref{sec:memreclamation} \crremove{and Section~\ref{appsec:memreclamation} in the Supplemental Materials} shows how to preserve this invariant when they are reclaimed.

\vspace{-2pt}
\subsection{Bundles and Contains Operations}
\label{sec:contains}
Contains operations must observe the same history as range queries in order to guarantee linearizability. 
The correctness anomaly that would arise if contains operations were to execute entirely without bundles is equivalent to the one reported in~\cite{adya1999weak, kishi2019sss, peluso2015gmu}, in which two update operations operating on different elements are observed in different order by concurrent reads.

Our solution is to treat the contains operation as a single-key range query. 
There is no additional overhead during the \texttt{pre-range} phase because the traversal is uninstrumented. 
Also, the \texttt{collect-range} phase will return the correct result since both the low and high keys are the same.
Like range queries, the \texttt{enter-range} phase preserves correctness because it enforces contains to be blocked if they reach a pending entry that is not finalized by a concurrent update.

One drawback of treating contains as a special case of range query is that reading \texttt{globalTs}
to set the linearization snapshot of a contains can increase contention on \texttt{globalTs}, significantly.
To eliminate this overhead, we observe that a contains operation \textit{does not} need to fix its linearization point by reading \texttt{globalTs}. Its linearization point can be delayed until dereferencing the last bundle in its traversal.
To implement that, contains calls \texttt{DereferenceBundle} with an infinitely large timestamp at Lines~\ref{line:enter-range} and~\ref{line:enter-range-while} instead of reading \texttt{globalTs} at Line~\ref{line:readglobalts}. 
In Section~\ref{sec:evaluation}, we assess the performance impact of this optimization.

\subsection{Correctness}

\crremove{Due to space constraints, we moved the detailed correctness proof to Section~\ref{appendix:correctness} of the Supplemental Materials, and we only include a brief intuition here. }
\cradd{The correctness proof can be found in the companion technical report~\cite{nelson2021bundling}, but we include a brief intuition here.}

The linearization point of an update operation is defined as the moment it increments \texttt{globalTs}; on the other hand, range queries are linearized when they read \texttt{globalTs} to set their snapshot. Considering that bundle entries are initialized in a pending state, this guarantees that range queries will observe all updates linearized before they read \texttt{globalTs}.

The linearization point of contains is less trivial. Briefly, it is dictated by the linearization point of the update operation that inserted the bundle entry corresponding to the last dereference made in the \texttt{enter-range} phase.
If this update increments \texttt{globalTS} concurrently with the contains, the contains is linearized immediately after the increment; otherwise, the contains is linearized at its start.

\cradd{\subsection{Applying Bundling}
Designing a bundled data structure is not meant to be an automated process; it requires selecting a concurrent linked data structure and engineering the deployment of the bundle abstraction.
That said, all the bundled data structures in this paper share common design principles.
First, they are all linked data structures and have a sentinel node to start from. 

Second, they are lock-based and all locks are acquired before modifications take place, and released after. 
This is better known as two-phase locking (2PL)~\cite{bernstein1983multiversion}.
Because the data structures follow 2PL, all bundle entries are finalized atomically for a given update operation, which makes reasoning about the correctness of the resulting bundled data structures easier.
In turn, this allows us to illustrate our technique with varying degrees of complexity (e.g., multiple children in the case of the binary search tree).
Note that this does not preclude other locking strategies, but care should be taken to ensure that bundles are updated appropriately.} 

\cradd{Finally, all of the data structures we apply bundling to implement non-blocking traversals that avoid checking locks. Although bundling makes reads blocking, it retains the optimized traversals during the \texttt{pre-range} phase and only synchronizes with updates (via the bundles) during the phases \texttt{enter-range} and \texttt{collect-range}. This is possible because we augment the existing links instead of replacing them.}

\section{Bundled Data Structures}


We now describe how to apply bundling to the highly-concurrent lazy sorted linked list~\cite{lazylist}. We also apply bundling to more practical data structures, such as the lazy skip list~\cite{lazyskiplist} and the Citrus unbalanced binary search tree~\cite{citrus}. 
\crremove{Due to space constraints, we only detail linked list in this section since it simplifies the presentation of bundling's application, and highlights the important aspects of the process.
We provide an overview here but include a closer discussion about bundled skip list and Citrus tree in Section~\ref{Appendix:other-ds} of the Supplemental Materials. }


\subsection{Bundled Linked List}
\label{sec:lazylist}

Listing~\ref{lst:lazylistnode} provides a full definition of member variables of the linked list node.
Note that the only additional field compared to the original algorithm is the bundle. 

\begin{algorithm}
\begin{scriptsize}
\KwIn{key, val}
\Begin{
    \While{true}{
       $pred, curr \leftarrow$ Traverse($key$) \\
       Lock($pred$) \\
       \label{line:lazylistinsertlock}
       \If{ValidateLinks($pred$, $curr$)}{
            \If{$curr.key == key$}{
                \Return false \\
            }
            $newNode \leftarrow$ \textbf{new} Node($key$, $val$, $curr$) \\
            {\color{blue}$bundles \leftarrow$ ($newNode.bundle$, $pred.bundle$) \\
            \label{line:bundles}
            $ptrs \leftarrow$ ($curr$, $newNode$) \\
            $ts \gets $PrepareBundles($bundles$, $ptrs$) \\
            \label{line:lazylistprepare}}
            $pred.next \gets newNode$\\ 
            \label{line:lazylistcritical}
            {\color{blue}FinalizeBundles($bundles$, $ts$)\\
            \label{line:lazylistfinalize}}
            Unlock($pred$) \\
            \Return true
       }
       Unlock($pred$) \\
   }
}
\caption{Insert operation of Bundled Linked List}
\label{algo:lazylistinsert}
\end{scriptsize}
\end{algorithm}


\underline{Insert Operation}. Initially, insert operations (Algorithm~\ref{algo:lazylistinsert}) traverse the data structure to determine where the new node must be added.
After locking the predecessor, both current and predecessor nodes are validated by checking that they are not logically deleted and that no node was inserted between them.
If validation succeeds and the key does not already exist, then a new node with its next pointer set to the successor is created. Otherwise, the nodes are unlocked and the operation restarts.

The above follow the same procedure that the original linked list (without bundling) would use.
\cradd{To illustrate how bundling applies, we highlight the bundling-specific lines blue.}
In case of successful validation, the next thing is to perform the three steps described in Section~\ref{sec:updates} to linearize an update operation in a bundled data structure: installing pending bundle entries and incrementing the global timestamp (Line~\ref{line:lazylistprepare}), performing the original critical section (Line~\ref{line:lazylistcritical}), and finalizing the bundles (Line~\ref{line:lazylistfinalize}).
For an insertion, the bundles of the newly added node and its predecessor must be modified to reflect their new values and the timestamp of the operation.
Then, locks are released.

Note that in Algorithm~\ref{algo:lazylistinsert} we employ an optimization where only the predecessor is locked by insert operations (Line~\ref{line:lazylistinsertlock}). In~\cite{lazylist}, it is proven that this optimization preserves linearizability.
However, it also reveals a subtle but important corner case that motivates the need for updates to wait for pending bundles to be finalized
(Line~\ref{line:updatewaitpendingbundle} of Algorithm~\ref{algo:prepare}).
Because the current node is not locked, it is possible that a concurrent update operation successfully locks the new node after it is reachable and before its bundles are finalized by the inserting operation.
\crremove{This nefarious case is protected by first waiting for the ongoing insertion to finish to ensure the bundle remains ordered (see Section~\ref{sec:updates}) or by holding a lock on the new node until it is inserted.}
\cradd{
This nefarious case is protected by first waiting for the ongoing insertion to finish to ensure the bundle remains ordered (see Section~\ref{sec:updates}).}

\cradd{
In Algorithm~\ref{algo:lazylistinsert}, the \texttt{PrepareBundles} step is invoked only after the locks are held.
Although this nesting synchronization might look redundant, it is important to guarantee safety in the aforementioned case.
Fusing the two is not out of the question, but would require careful consideration regarding both updates and range queries.
Informally, a range query could first read the timestamp then check the lock, waiting until it is released by the update, akin to the pending entry in the original algorithm. 
We mention this point to illustrate that bundling is adaptable, but we do not include a design since it tightly couples our approach to the underlying data structure and lengthens the critical section during which a read may have to wait.
}

\underline{Remove operations}. Remove operations follow a similar pattern by first traversing to the appropriate location, locking the nodes of interest, validating them (restarting if validation fails), marking the target node as deleted, unlinking it if its key matches the target key, then finally unlocking the nodes and returning.
Here, the original critical section includes the logical deletion and the unlinking of the deleted node, which is therefore surrounded with the required bundle maintenance by performing the operation via calls to \texttt{PrepareBundles} and \texttt{FinalizeBundles}.

Importantly, range query operations do not need to check if a node is logically deleted because they are linearized upon reading the global timestamp and all bundle entries point to nodes that existed in the data structure at the corresponding timestamp.
The removal of a node is observed through the finalization of the predecessor's bundle if their linearization point occurs after the remove operation increments \texttt{globalTs}.
If the range query's linearization point falls before the remove, then the newly inserted entry in the predecessor's bundle will not satisfy the observed \texttt{globalTs} value and an older entry will be traversed instead.

Similarly, contains do not check for logical deletions because they are linearized according to the last entry traversed during the \texttt{enter-range} phase and will be linearized after a node is physically unlinked, or before it is logically deleted.

\crremove{
Note that a removed node's bundle will not change because its \texttt{ptr} value reflects the physical state of the data structure immediately before the removal takes place.
However, the logical deletion is still required by update operations to validate that no concurrent operation removes the predecessor node.}


\cradd{
The removed node's bundle is also updated and points to the head of the list. 
This is to protect against a rare case resulting from the \texttt{pre-range} phase of a concurrent range query returning the removed node. In this case, if the range query reads \texttt{globalTs} (i.e., is linearized) after a sequence of deletions of all keys starting from this node through some nodes in the range, the range query will mistakenly include the deleted nodes.
Adding a bundle entry in the removed node solves this problem by redirecting the range query to the head node and allowing the \texttt{enter-range} phase to traverse from the beginning using the bundle entries matching the \texttt{globalTs} value it reads.
As shown in Section~\ref{sec:evaluation}, this scenario rarely occurs and has no impact on performance.


}

\underline{Read operations}. We now analyze the functions required by contains and range queries: \texttt{GetNext}, \texttt{GetNextFromBundle}, and \texttt{CollectRange}.
To do so, we describe the execution of a range query and then discuss the changes necessary to support contains operations.

When executing a range query on the linked list, the \texttt{pre-range} phase consists of scanning from the head until the node immediately preceding the range (or target key) is found using \texttt{GetNext}, which simply returns \texttt{next}. 
The operation then traverses using the bundles until finding the first node in the range using \texttt{GetNextFromBundle}. This function uses a single call to \texttt{DereferenceBundle} to return the node that satisfies \texttt{ts} in the bundle of the current node.
While the current node $curr$ is in the range, \texttt{CollectRange} adds $curr$ to the result set then updates it to point to the next node in the snapshot calling \texttt{DereferenceBundle} on \texttt{bundle}.
Because of the construction of a linked list, the operation terminates once the range is exceeded.
Together, these functions are used by Algorithm~\ref{algo:rangequery}
to perform linearizable read operations.

Recall from Section~\ref{sec:contains} that contains operations follow the same general strategy as range queries, but use a infinitely large timestamp to traverse the newest bundle entry at each node.
Hence, the value stored at Line~\ref{line:readglobalts} of Algorithm~\ref{algo:rangequery} can be replaced with the maximum timestamp.

\subsection{Bundled Skiplist}
\label{sec:skiplist}

The second data structure where we apply bundling is the lazy skip list~\cite{lazyskiplist}, whose design is similar to the bundled linked list.
In the following, we highlight the differences between the two designs. 

The first difference is that skip list consists of a bottom \textit{data layer} where data resides, and a set of \textit{index layers} to accelerate traversal. 
Hence, given a target key, a regular traversal returns a set of (\texttt{pred}, \texttt{curr}) pairs at both the index and data layers.
If the target key exists in the data structure, then it also returns the highest level (\texttt{levelFound}) at which the node was found. 
A naive approach to bundling this skip list would be to replace all links with bundled references, including the index layers.
However, recall that for read operations the use of bundles is delayed to improve performance. 
Therefore it is sufficient to only bundle references at the data layer, leaving the index layers as is.


Second, because update operations manipulate multiple links per node, they are linearized using logical flags.
Specifically, insert operations set a \texttt{fullyLinked} flag in the new node after the links of all its \texttt{pred} nodes are updated to point to it. 
Setting this flag is the linearization point of insert operations in the original lazy skip list. 
As required by bundling, the original critical section that sets the predecessors' links and marks the node as logically inserted is book-ended by the preparation (Algorithm~\ref{algo:prepare}) and finalization (Algorithm~\ref{algo:finalize}) of the bundles for the predecessor and the new node, similarly to the bundled linked list.


Similar to inserts use of \texttt{fullyLinked}, remove operations are linearized in the original lazy skip list by a logical deletion, handled in the bundled skip list as follows.
First, the \cradd{bundles of the predecessor and the removed node} are prepared using Algorithm~\ref{algo:prepare}.
\cradd{Similar to the lazy list, the removed node's bundle will point to the head of the list.}
Next, the critical section marks the node as logically deleted then updates the references of the predecessor nodes in the index and data layers, to physically unlink the node.
Finally, \cradd{the bundles are finalized} with Algorithm~\ref{algo:finalize}, allowing read operations to observe the change.



Finally, to support linearizable read operations, the skip list defines the required functions as follows.
\texttt{GetNext} leverages the index layers to find the node whose next node at the data layer is greater than or equal to the target key by examining.
Similar to the bundled linked list, \texttt{GetNextFromBundle} returns the bundle entry satisfying $ts$, but this will only traverse the data layer.
\texttt{CollectRange} then scans the data layer, using bundles, to collects the range query's result set.



\subsection{Bundled Binary Search Tree}
\label{sec:citrus}

For our bundled tree, we reference the Citrus unbalanced binary search tree~\cite{citrus}, which leverages RCU and lazy fine-grained locking to synchronize update operations while supporting optimistic traversals.
We modify it by replacing each child link of the search tree with a bundled reference.


Citrus implements a traversal enclosed in a critical section protected by RCU's read lock.
This protects concurrent updates from overwriting nodes required by the traversal.
After the traversal, if the current node matches the target, the operation returns a reference to it (named \texttt{curr}), a reference to 
its parent (named \texttt{pred})
and the direction of the child from \texttt{pred}.
Otherwise, the node is not found and the return value of \texttt{curr} is null.

Because the Citrus tree is unbalanced, insertions are straightforward and always insert a leaf node. 
Otherwise adhering to the original tree algorithm, insertions are linearized by first preparing the bundle of the new node and of \texttt{pred}, corresponding to the child updated and incrementing the global timestamp with \texttt{PrepareBundles}, then setting the appropriate child, and finalizing the bundles by calling \texttt{FinalizeBundles}.
Lastly, the insert operation unlocks \texttt{pred} and returns.

The more interesting case is a remove operation, which should address three potential situations, assuming that the target node \texttt{curr} is found and will be removed.
In the first case, \texttt{curr} has no children.
To remove \texttt{curr} the child of \texttt{pred} that pointed to \texttt{curr} is updated along with its bundle.
In the second case, \texttt{curr} has exactly one child.
In this scenario, the only child of \texttt{curr} replaces \texttt{curr} as the child of \texttt{pred}.
Again, the bundle corresponding to \texttt{pred}'s child is also updated accordingly.
The last, and more subtle, case is when \texttt{curr} has two children, in which we must replace the removed node with its successor (the left-most node in its right subtree).
\cradd{In all cases, both of the removed node's bundles are updated to point to the root to avoid the divergent path scenario described in Section~\ref{sec:lazylist}.}

In this last case,
both \texttt{curr}'s successor and its parent are locked. 
Then, following RCU's methodology, a copy of the successor node is created and initialized in a locked state with its children set to \texttt{curr}'s children.
The effect of this behavior
is that possibly six bundles must be modified during \texttt{PrepareBundles} and \texttt{FinalizeBundles} to reflect the new physical state after the operation takes effect.
\cradd{The required changes are: 1) \texttt{pred}'s left or right bundle is modified with a reference to the copy of \texttt{curr}'s successor, 2) both bundles in the copy of \texttt{curr}'s successor must point to the respective children of \texttt{curr}, and 3) both of \texttt{curr}'s bundles must lead to the root of the tree to avoid the divergent path scenario described in Section~\ref{sec:lazylist}}.
Optionally, if the parent of \texttt{curr}'s successor is not \texttt{curr} itself, then this node's bundle is updated to point to the successor's right-hand branch, since the successor is being moved.
In all cases, the remove operation is linearized at the moment the child in \texttt{pred} is changed, making the update visible.


Read operations differ slightly from the bundled linked list and skip list implementations. For trees, unlike lists, the node found during the \texttt{pre-range} phase is not necessarily a node whose key is lower than the lower bound of the range.
Instead, it is the first node discovered through a traversal whose child is in the range.
Its child is the root of the sub-tree that includes all nodes belonging to the range.
More concretely, \texttt{GetNext} examines the current node and either follow the left or right child depending on whether the node is greater than or less than the range, respectively.
Similar to before, the node reached by the optimistic traversal of the \texttt{pre-range} phase may not be the correct entry point to the range, and subsequent \texttt{enter-range} phase using bundles is needed.
\texttt{GetNextFromBundle} parallels \texttt{GetNext}, but uses bundles instead of child pointers.


As in~\cite{ebr-rq}, \texttt{CollectRange} follows a depth-first traversal using bundles.
A stack is used to keep the set of to-be-visited-nodes to help traverse the subtree rooted at the node returned by the \texttt{enter-range} phase.
\texttt{CollectRange} initializes the stack with the first node reached in the range.
It then iteratively pops a node from the stack and checks whether its key is lower than, within, or greater than the range.
Next, it adds the node's corresponding children to the stack according to this check.
If the node is within the range it adds the node to the result set.
Finally, it pops the next node from the stack and performs the above procedure again.
Note that, as previously discussed, contains will follow the same general pattern as range queries but will use an infinitely large timestamp to traverse the newest entry at each link.

\section{Memory Reclamation}
\label{sec:memreclamation}

We rely on epoch-base memory reclamation (EBR)~\cite{ebr} to cleanup physically removed nodes and no longer needed bundle entries because, as already assessed by~\cite{ebr-rq}, quiescent state memory reclamation~\cite{rcu} (a generalized form of EBR) mirrors the need for a range query to observe a snapshot of the data structure.
To avoid overrunning memory with bundles, we employ a background cleanup mechanism to recycle outdated bundle entries.
This is accomplished by tracking the oldest active range query and only retiring bundle entries that are no longer needed by any active operation.


\textbf{EBR Overview}
EBR guarantees that unreachable objects are freed by maintaining a collection of references to recently retired objects and an epoch counter to identify. It operates under the assumption that object are not accessed outside of the scope of an operation (i.e., during quiescence). EBR monitors the epoch observed by each thread and the objects retired during each epoch. The epoch is only incremented after all active threads have announced that they have observed the current epoch value. When a new epoch is started, any objects retired two epochs prior can be safely freed.
We use a variant of EBR, called DEBRA~\cite{debra}, that stores per-thread limbo lists which also reduces contention on shared resources by recording removed nodes locally for each thread.

\textbf{Freeing Data Structure Nodes.} 
EBR guarantees that no node is freed while concurrent range queries (as well as any concurrent primitive operation) may access it; 
and, bundling guarantees that no range query that starts after physically removing a node will traverse to this node.
As an example, consider the two following operations: \textit{i}) a range query, $R$, whose range includes node $x$; and \textit{ii}) a removal operation, $U_t$, which is linearized at time $t$ and removes $x$.
If $R$ is concurrent with $U_t$, then EBR will guarantee that $x$ is not freed since $R$ was not in a quiescent state and a grace period has not passed.
In this case, $R$ may safely traverse to $x$ based on its observed timestamp, without concern that the node may be freed.
On the other hand, if $R$ starts after $U_t$, then trivially $x$ will never be referenced by $R$ and is safe to be reclaimed since $R$ observed a timestamp greater than or equal to $t$.

\textbf{Freeing Bundle Entries.} Bundle entries are reclaimed in two cases. The first, trivial, case is that bundle entries are reclaimed when a node is reclaimed. The second case is more subtle. After a node is freed, there may still exist references to it (in other nodes' bundles) that are no longer necessary and should be freed.
Bundle entries that have a timestamp older than the oldest active range query can be reclaimed only if there also exists a more recent bundle entry that satisfies the oldest range query.
This cleanup process may be performed during operations themselves or, as we implement, delegated to a background thread.

To keep track of active range queries, we augment the global metadata with \texttt{activeRqTsArray}, which is an array of timestamps that maintains their respective linearization timestamp. 
During cleanup, this array is scanned and the oldest timestamp is used to retire outdated bundle entries. Once a node is retired, EBR becomes responsible for reclaiming its memory in the proper epoch.

Reading the global timestamp and setting the corresponding slot in \texttt{activeRqTsArray} must happen atomically to ensure that a snapshot of the array does not miss a range query that has read the global timestamp but not yet announced its value.
This is achieved by first setting the slot to a pending state, similar to they way we protect bundle entries, which blocks the cleanup procedure until the range query announces its linearization timestamp.
Second, the cleanup thread has to be protected by EBR as well. 

Since contains operations use the most recent bundle entry at each node, they do not need to announce themselves in \texttt{activeRqTsArray} since we guarantee that the background cleanup thread never retires the newest node. Note that a contains operation, as well as all other operations, still announces the EBR epoch it observes when it starts as usual. Accordingly, even if the entry it holds becomes no longer the most recent entry and is retired, it will not be reclaimed until the contains operation finishes.

\section{Evaluation}
\label{sec:evaluation}

\crremove{To evaluate our design, we compare} \cradd{We evaluate} bundling \cradd{by integrating our design into an existing benchmark~\cite{ebr-rq} and comparing} with the following competitors: \textit{EBR-RQ}~\cite{ebr-rq} is the lock-free variant of Arbel-Raviv and Brown's technique based on epoch-based reclamation; \textit{RLU}~\cite{rlu} implements an RCU-like synchronization mechanism for concurrent writers; \textit{vCAS}~\cite{vcas-ppopp21} is our porting of the target data structures to use vCAS objects; \textit{Unsafe} is an implementation of each data structure for which range queries are uninstrumented (i.e., non-linearizable).

It should be noted that although the infrastructure supporting range queries in EBR-RQ is lock-free, the data structures still require locking.
We ignore the locking variant of EBR-RQ because it consistently performs worse at all configurations.
For the same reason, we do not include Snapcollector~\cite{snapcollector} in our plots.
Also, RLU is not included in the results for skip list because no implementation is available.

Since no procedural conversion for lock-based data structures to a vCAS implementation exists, we follow their guideline by replacing with a vCAS object every pointer and metadata necessary for linearization, except for the locks.
During updates, a new version is installed in the vCAS object whenever the field changes by using the vCAS API.
Since the mutable fields are vCAS objects and all accesses to those fields are instrumented, operations are linearizable.
\cradd{Although it is not clear whether this porting can be generalized to any lock-based data structure, we believe it is sufficient to contrast performance with our approach.}

In all experiments, we added another version of Bundle, named \textit{Bundle-RQ}, which moves the responsibility of updating \texttt{globalTs} to range queries (similar to~\cite{vcas-ppopp21} and~\cite{ebr-rq}). A general observation is that there is no significant difference in performance between Bundle and Bundle-RQ throughout all experiments. 
\cradd{Incrementing the timestamp on updates has as good or better performance, except in the 100\% update workload where it is slightly worse. As a result, our decision to increment the timestamp on updates is not a dominating factor in performance. The timestamp itself does pose a scalability bottleneck, but this is observed across all competitors since they all use logical timestamps to represent versions.}
Because of this, we avoid discussing the details of this version for brevity and we focus on commenting Bundle's performance in the rest of the evaluation.
We also evaluated a version in which the contains operation reads \texttt{globalTs} instead of calling \texttt{DereferenceBundle} with an infinitely large timestamp. However, this version was constantly 10\%-20\% worse in light update workloads than our optimized Bundle version, due to the increasing cache contention on \texttt{globalTs}. Accordingly, we excluded it from the figures for clarity.

\cradd{Finally, recollect that removals in our bundled data structures insert entries in the bundles of removed nodes pointing to the head (or root for the Citrus tree). Our experiments reveal no impact on the overall performance since traversals of this kind occur at most four times for every 1 million operations, across all data structures and workloads.}

All code is written in C++ and compiled with \texttt{-std=c++11 -O3 -mcx16}.
Tests are performed on a machine with four Xeon Platinum 8160 processors, with a total of 192 hyper-threaded cores \cradd{and a 138 MB L3 cache}, running Ubuntu 20.04.
All competitors leverage epoch-based memory reclamation to manage memory.
EBR-RQ accesses the epoch-based memory reclamation internals to provide linearizable range queries.
All other competitors have strategies that are orthogonal to how memory is managed.
While memory is reclaimed, each approach avoids freeing memory to the operating system to eliminate performance dependence on the implementation of the \texttt{free} system call.
The results for bundling include the overhead for tracking the oldest active range query for cleaning up outdated bundle entries.

\subsection{Bundled Data Structure Performance}
\label{sec:microbench}
In each of the following experiments, threads execute a given mix of update, contains, and range query operations, with target keys procured uniformly. 
The data structure is initialized with half of the keys in the key range; all updates are evenly split between inserts and removes stabilizing the data structure size at half-full.
Workloads are reported as $U-C-RQ$, where $U$ is the percentage of updates, $C$ is the percentage of contains and $RQ$ is the percentage of range queries.
All reported results are an average of three runs of three seconds each, except where noted.
The key range of each data structure is as follows: the lazy list is 10,000 and the skip list and Citrus tree are both 1,000,000. \cradd{Except where noted, ranges are 50 keys long with starting points uniformly generated from the set of possible keys.}

\begin{figure}[h]
\centering
\begin{subfigure}{\linewidth}
        \centering
        \includegraphics[width=\textwidth]{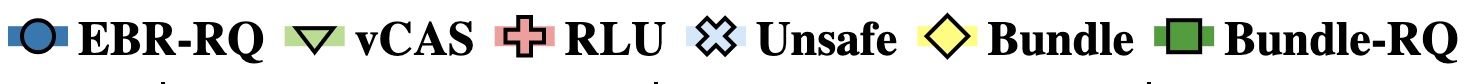}
\end{subfigure}
\begin{subfigure}{.45\linewidth}
    \centering
    \includegraphics[width=\textwidth]{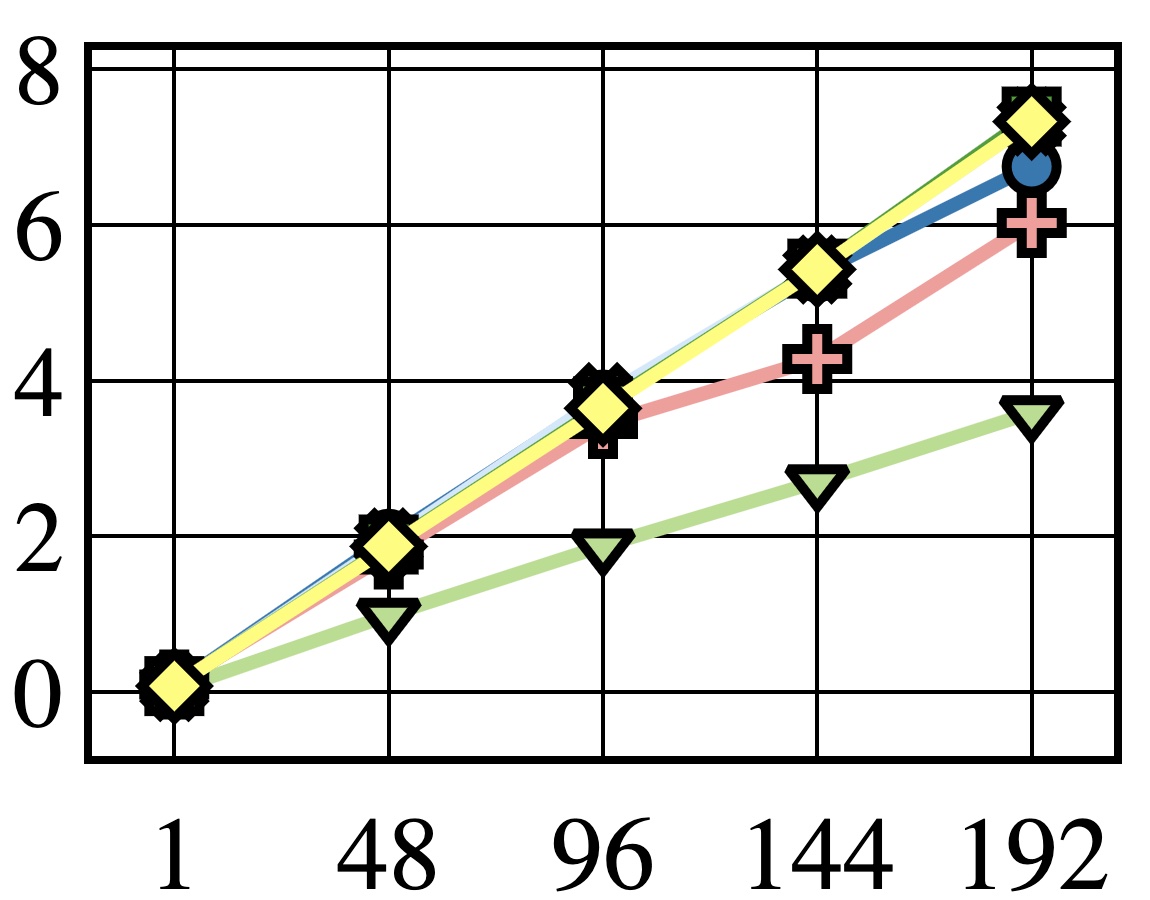}
\end{subfigure}
\begin{subfigure}{.45\linewidth}
    \centering
    \includegraphics[width=\textwidth]{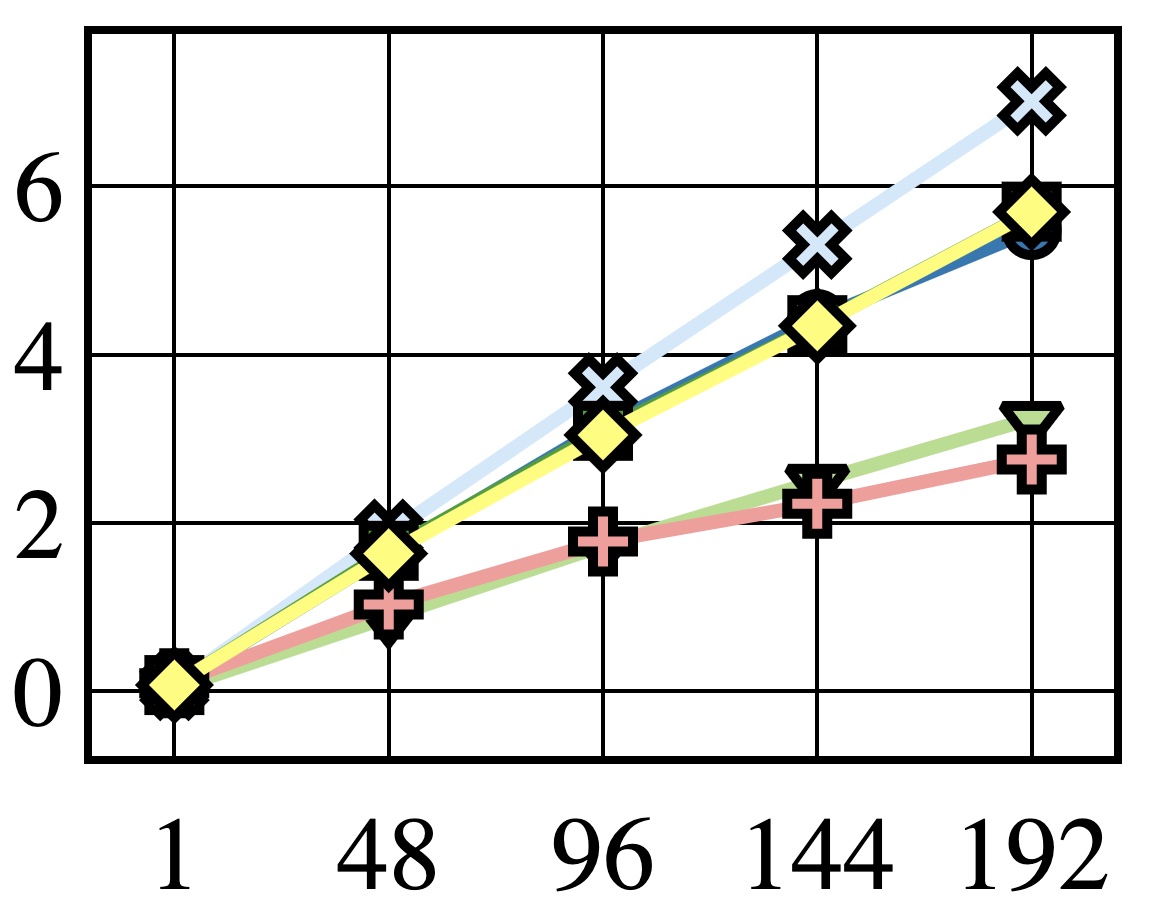}
\end{subfigure}
    \caption{Throughput (Mops/s) of a lazy list for  10-80-10 and 90-0-10 workloads \cradd{while varying the number of threads.}}
    \label{fig:lazylist}
\end{figure}


\begin{figure*}[t]
    \centering
    \begin{subfigure}{\textwidth}
        \centering
        \includegraphics[width=.6\textwidth]{figures/legend-microbench.jpg}
    \end{subfigure}\\
    \hspace{-15pt}
    \begin{subfigure}{0.025\textwidth}
        \includegraphics[width=1.4\textwidth]{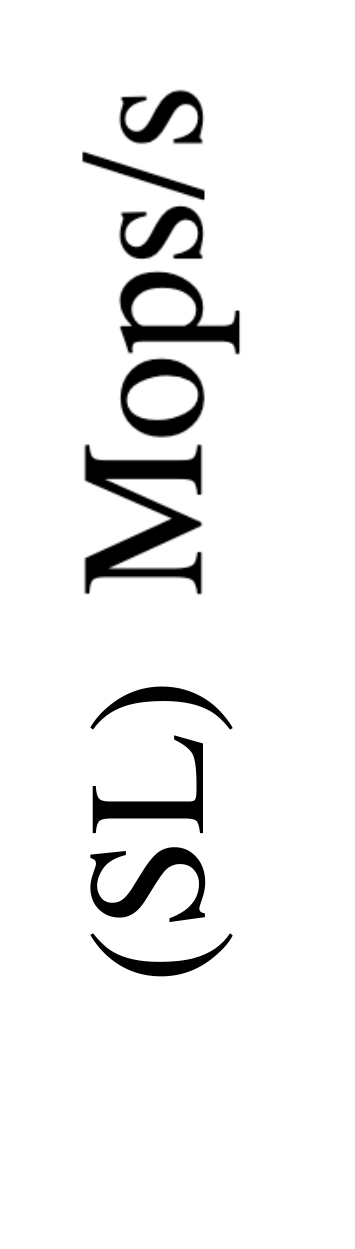}
    \end{subfigure}
    \hspace{-6pt}
    \begin{subfigure}{0.17\textwidth}
        \includegraphics[width=\textwidth]{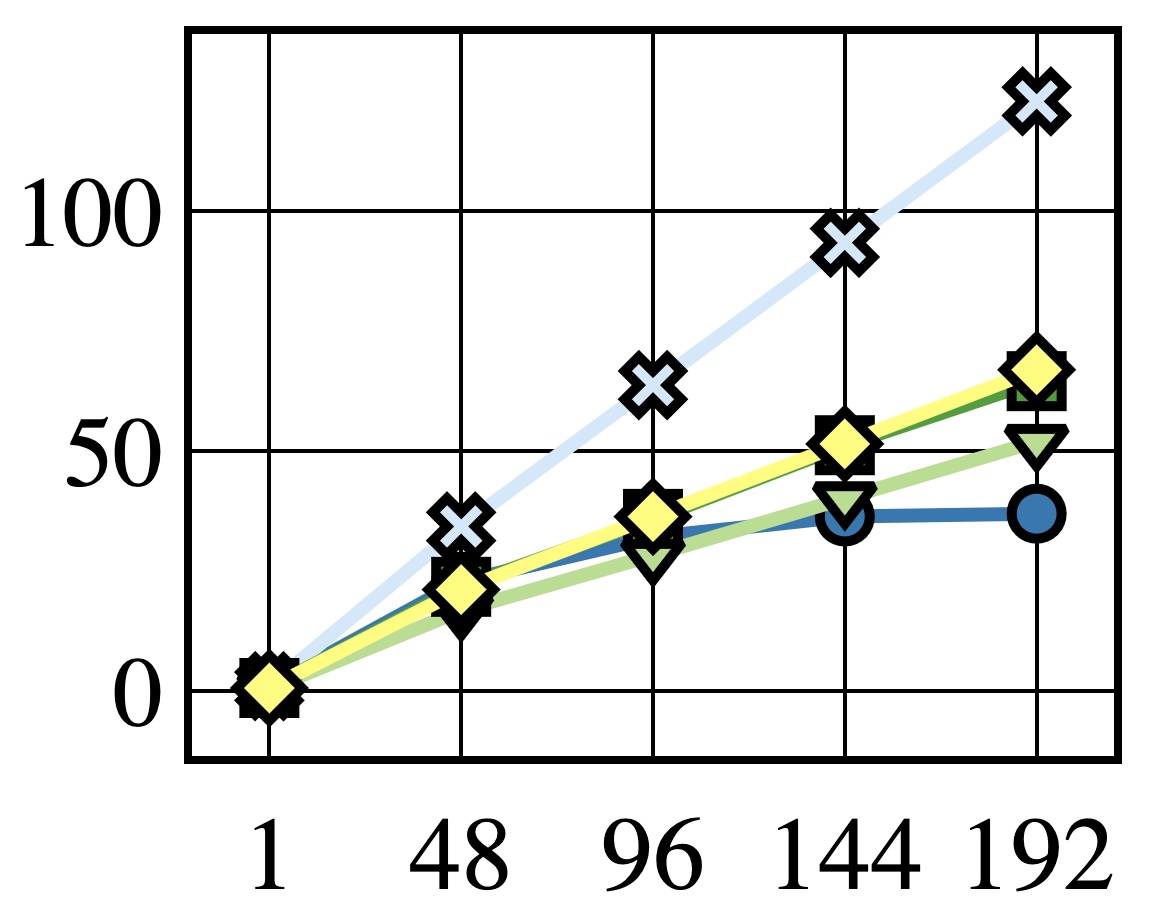}
        \caption{$2-88-10$}
        \label{fig2:sl-2-88-10}
    \end{subfigure}
    \hspace{-7pt}
    \begin{subfigure}{0.17\textwidth}
        \includegraphics[width=\textwidth]{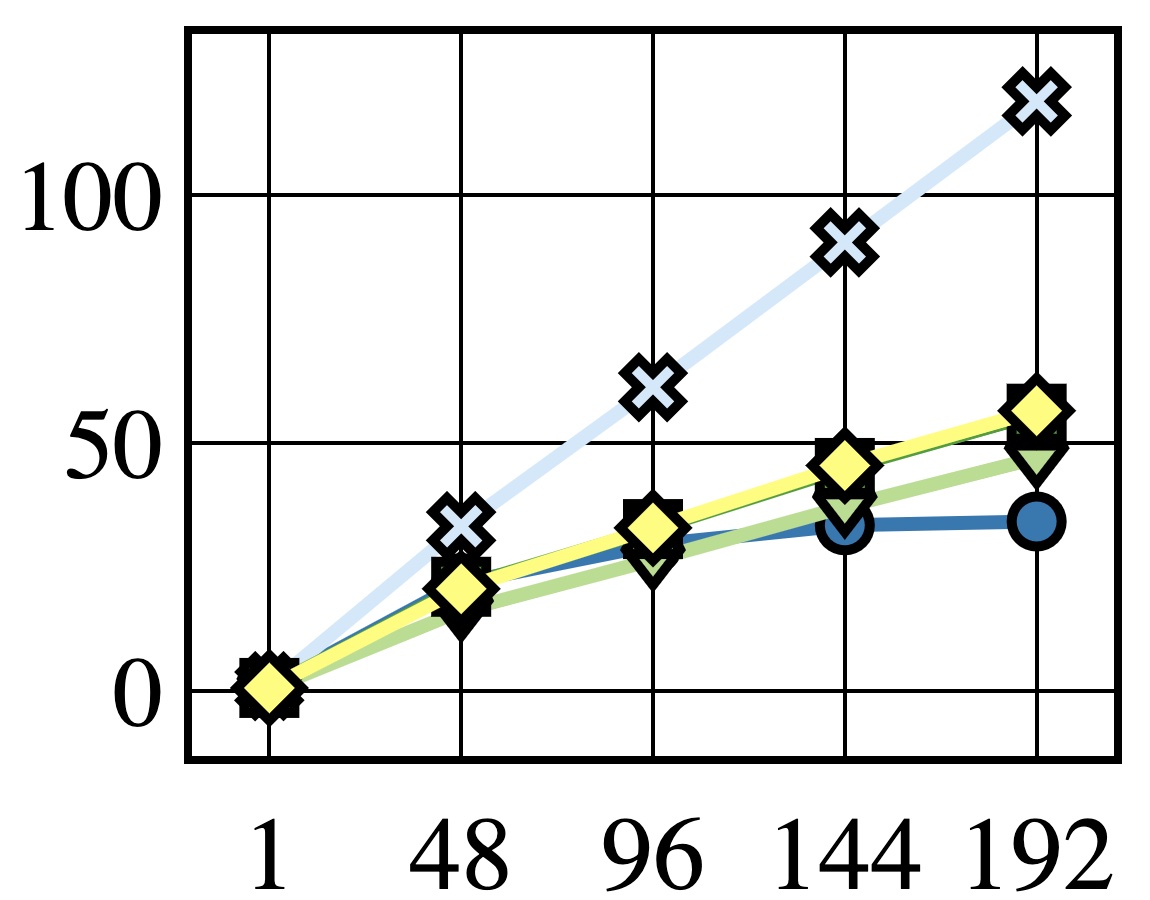}
        \caption{$10-80-10$}
        \label{fig2:sl-10-80-10}
    \end{subfigure}
    \hspace{-7pt}
    \begin{subfigure}{0.17\textwidth}
        \includegraphics[width=\textwidth]{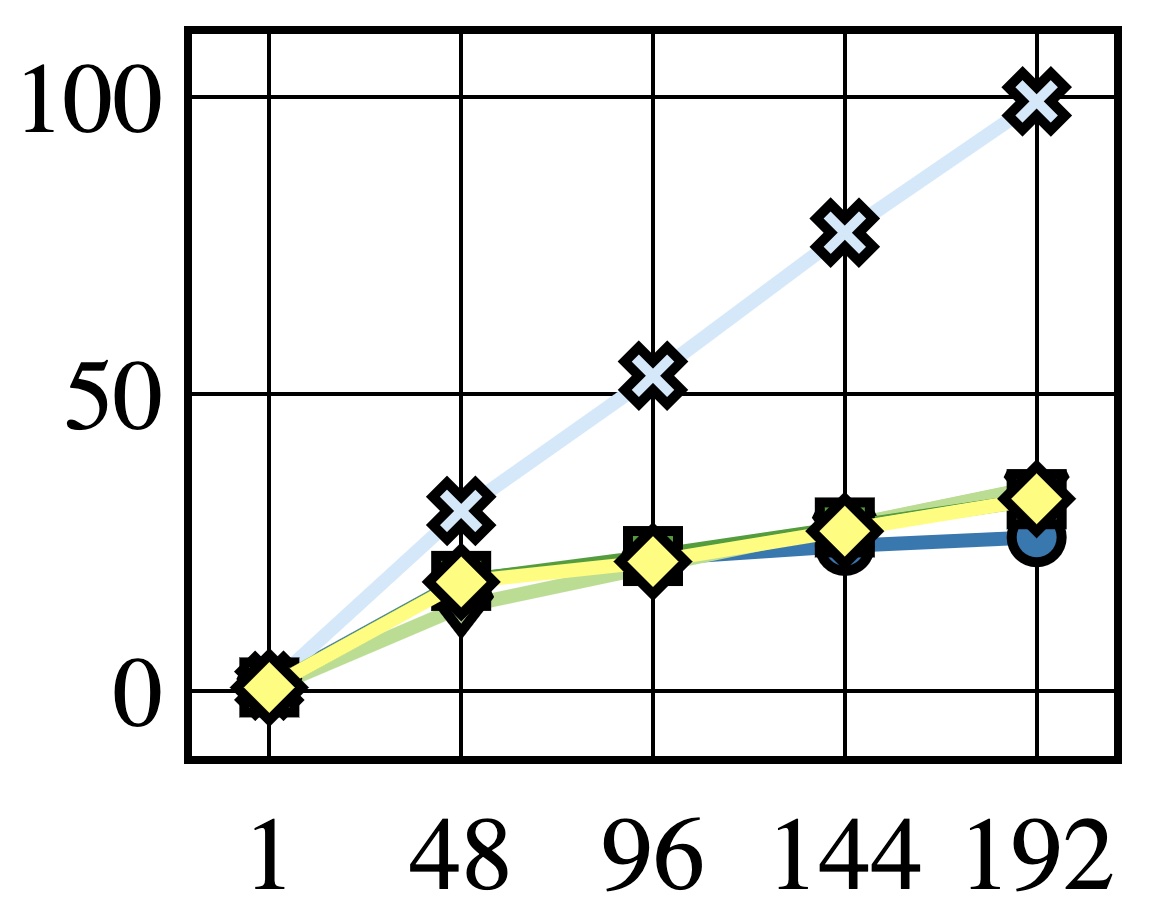}
        \caption{$50-40-10$}
        \label{fig2:sl-50-40-10}
    \end{subfigure}
    \hspace{-7pt}
    \begin{subfigure}{0.17\textwidth}
        \includegraphics[width=\textwidth]{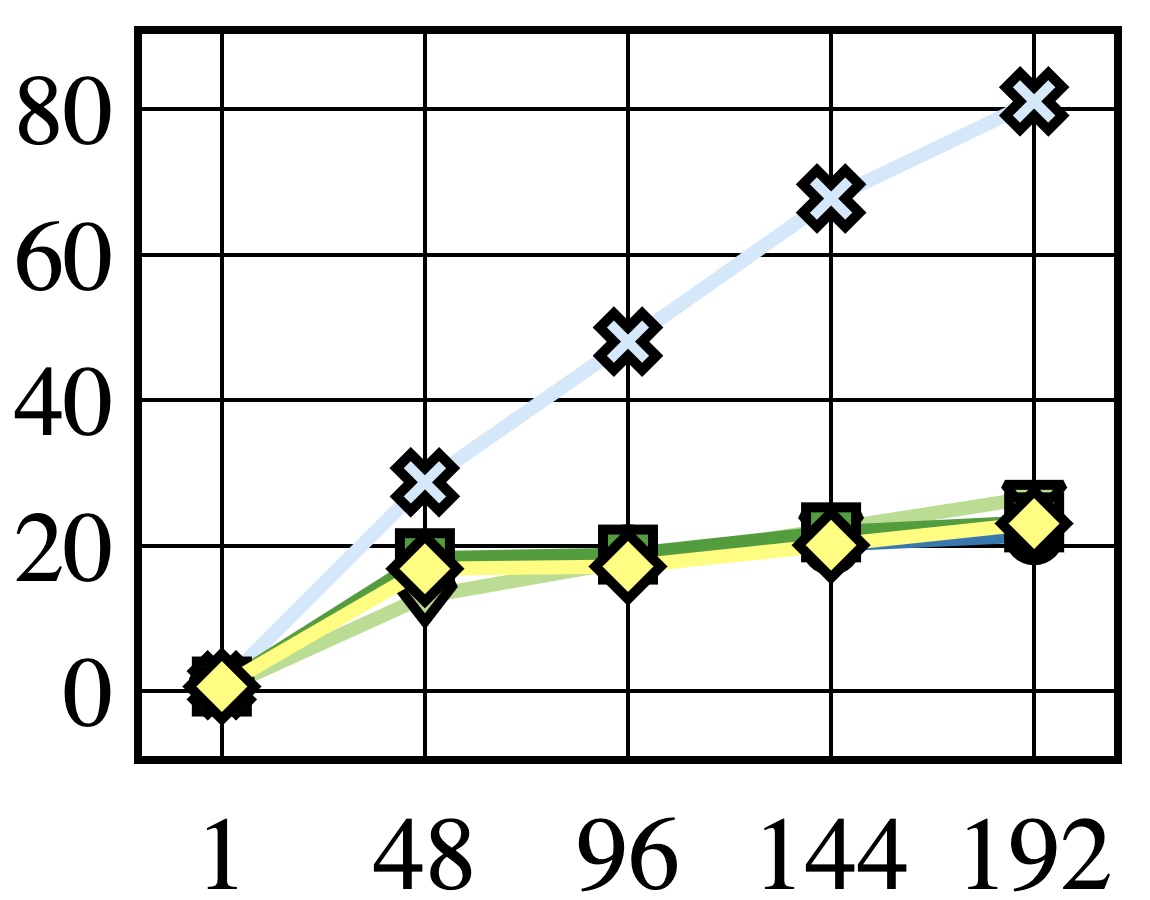}
        \caption{$90-0-10$}
        \label{fig2:sl-90-0-10}
    \end{subfigure}
    \hspace{-7pt}
    \begin{subfigure}{0.17\textwidth}
        \includegraphics[width=\textwidth]{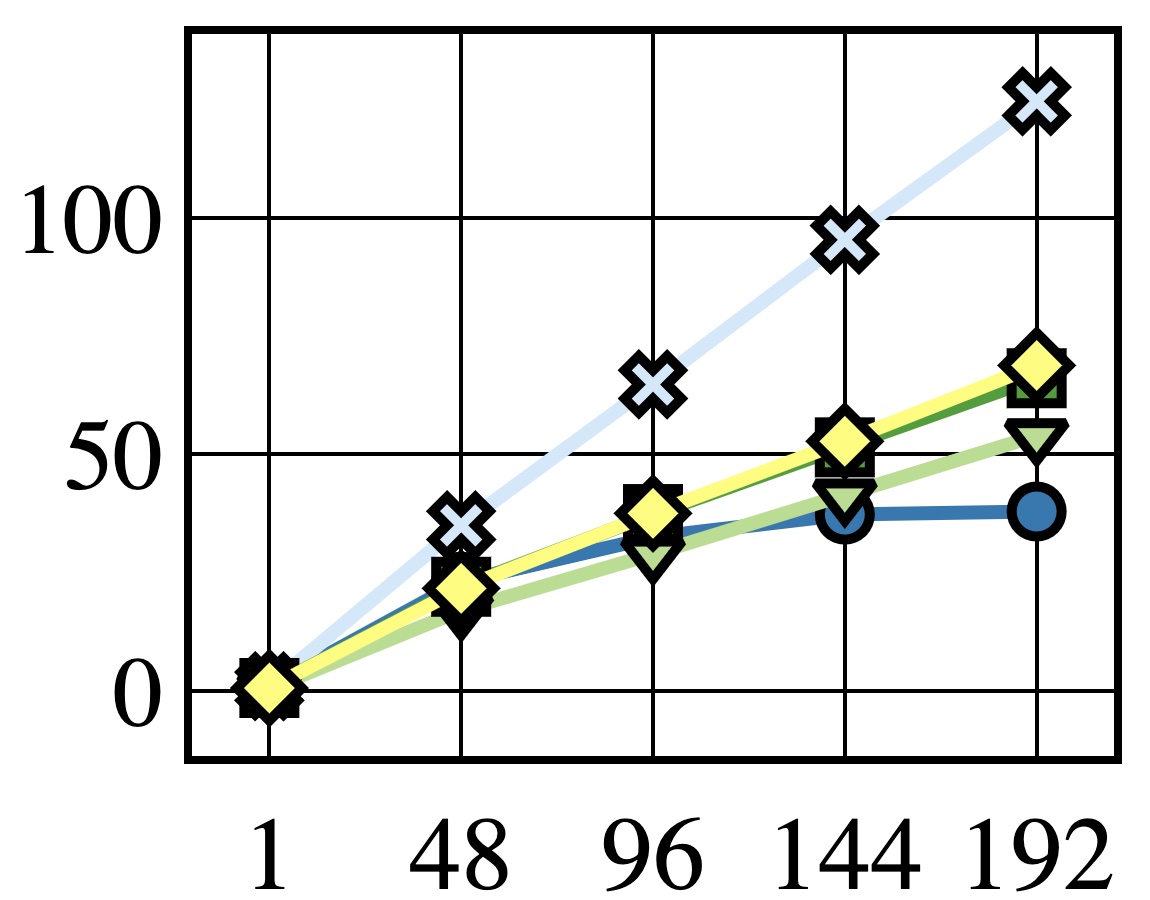}
        \caption{$0-90-10$}
        \label{fig2:sl-0-90-10}
    \end{subfigure}
    \hspace{-7pt}
    \begin{subfigure}{0.17\textwidth}
        \includegraphics[width=\textwidth]{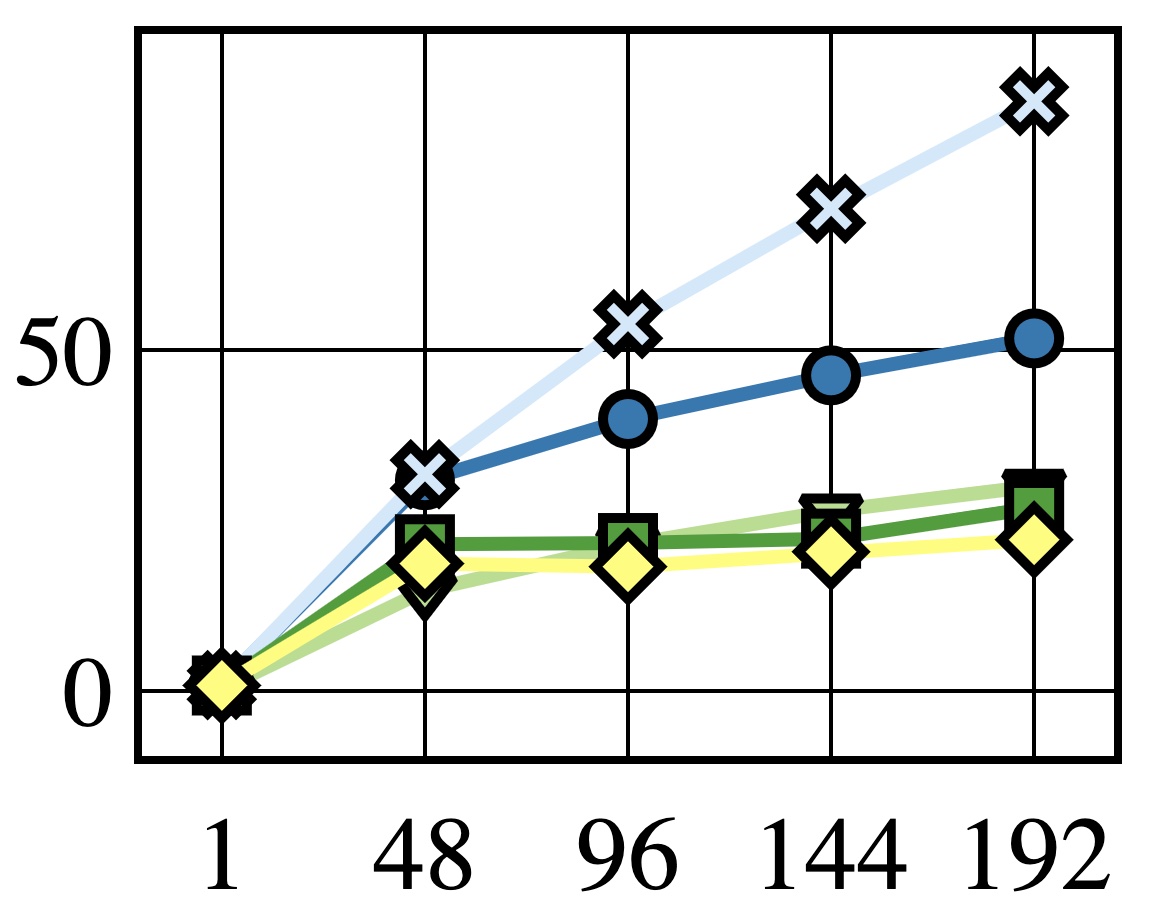}
        \caption{$100-0-0$}
        \label{fig2:sl-100-0-0}
    \end{subfigure}\\
    \hspace{-15pt}
    \begin{subfigure}{0.025\textwidth}
        \vspace{-15pt}
        \includegraphics[width=1.4\textwidth]{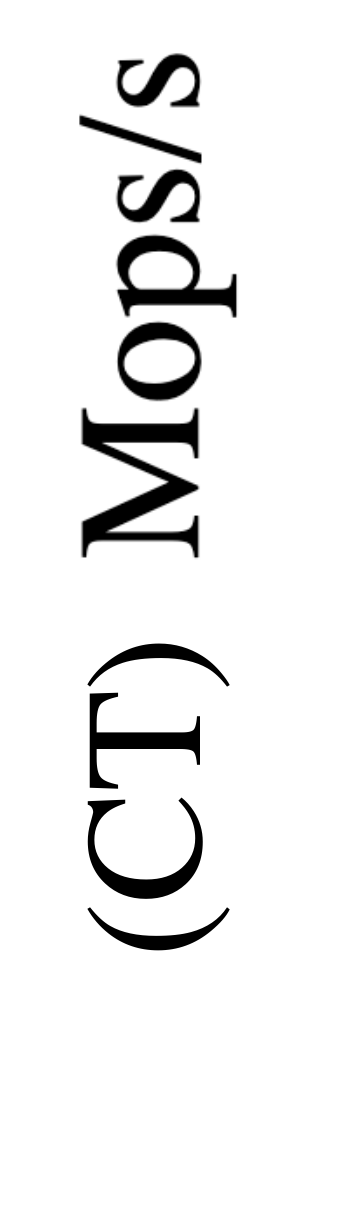}
    \end{subfigure}
    \hspace{-6pt}
    \begin{subfigure}{0.17\textwidth}
        \includegraphics[width=\textwidth]{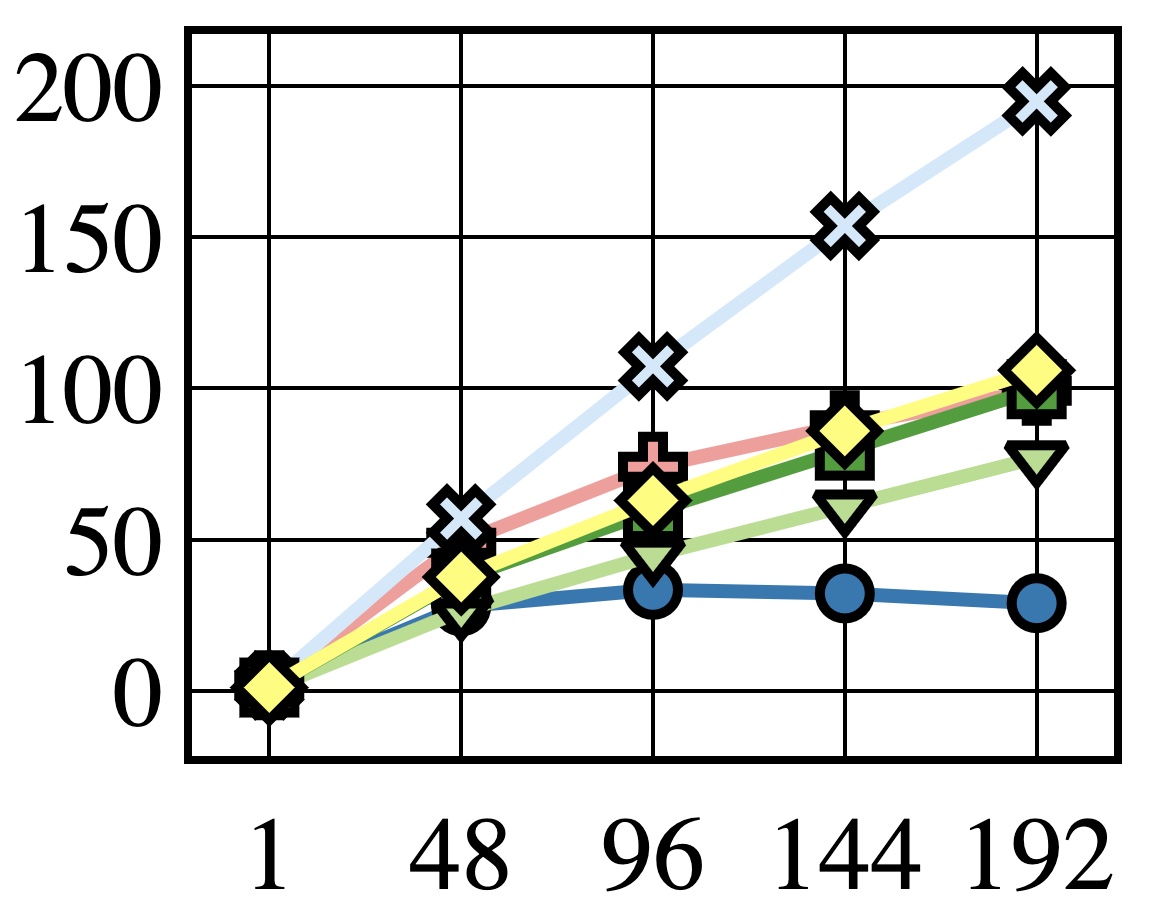}
        \caption{$2-88-10$}
        \label{fig2:ct-2-88-10}
    \end{subfigure}
    \hspace{-7pt}
    \begin{subfigure}{0.17\textwidth}
        \includegraphics[width=\textwidth]{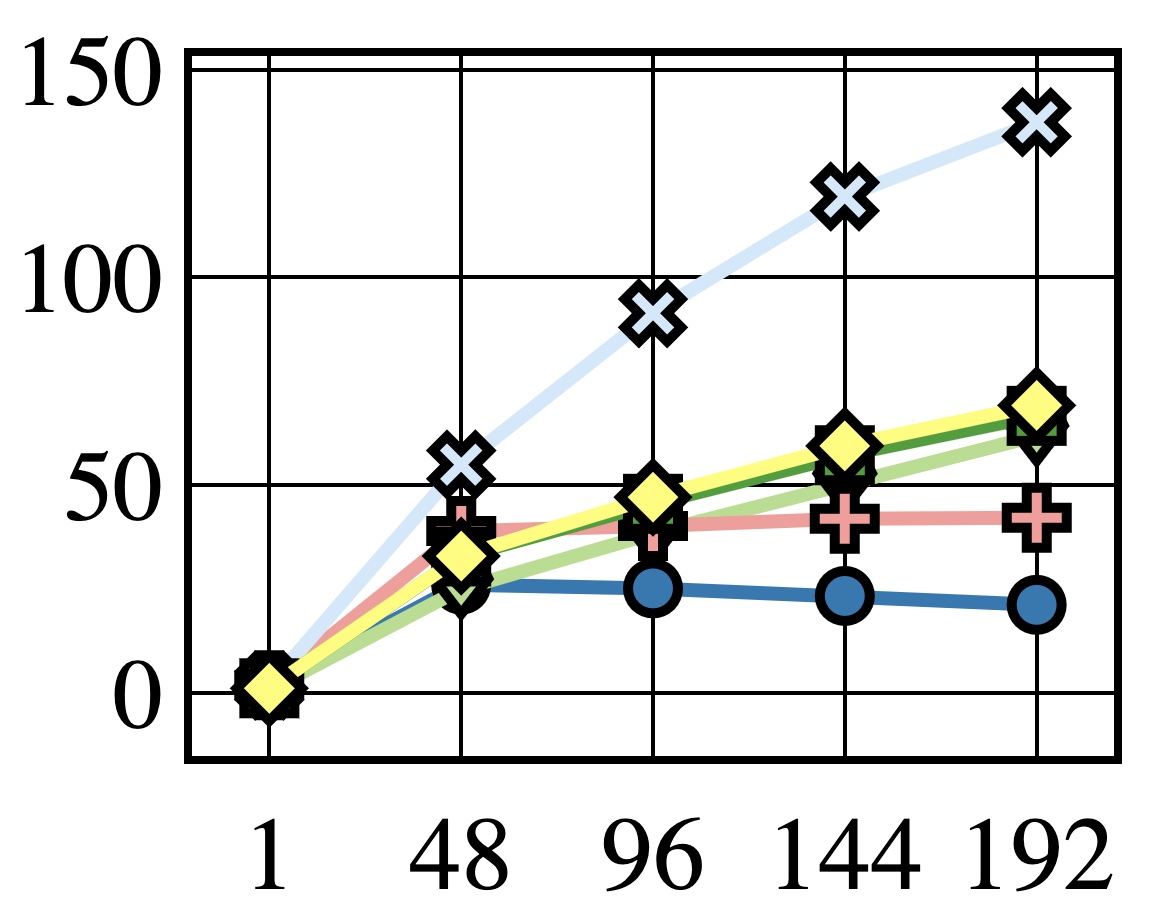}
        \caption{$10-80-10$}
        \label{fig2:ct-10-80-10}
    \end{subfigure}
    \hspace{-5pt}
    \begin{subfigure}{0.17\textwidth}
        \includegraphics[width=\textwidth]{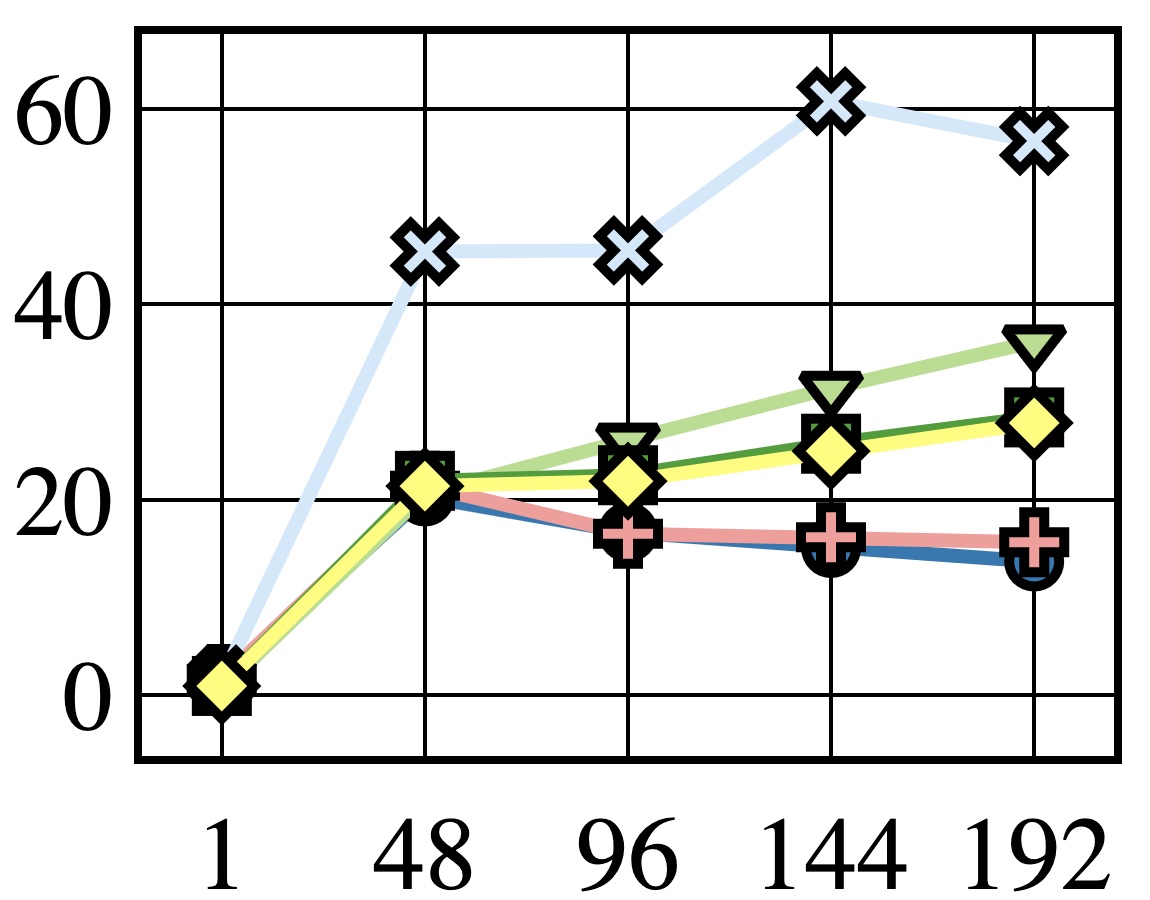}
        \caption{$50-40-10$}
        \label{fig2:ct-50-40-10}
    \end{subfigure}
    \hspace{-7pt}
    \begin{subfigure}{0.17\textwidth}
        \includegraphics[width=\textwidth]{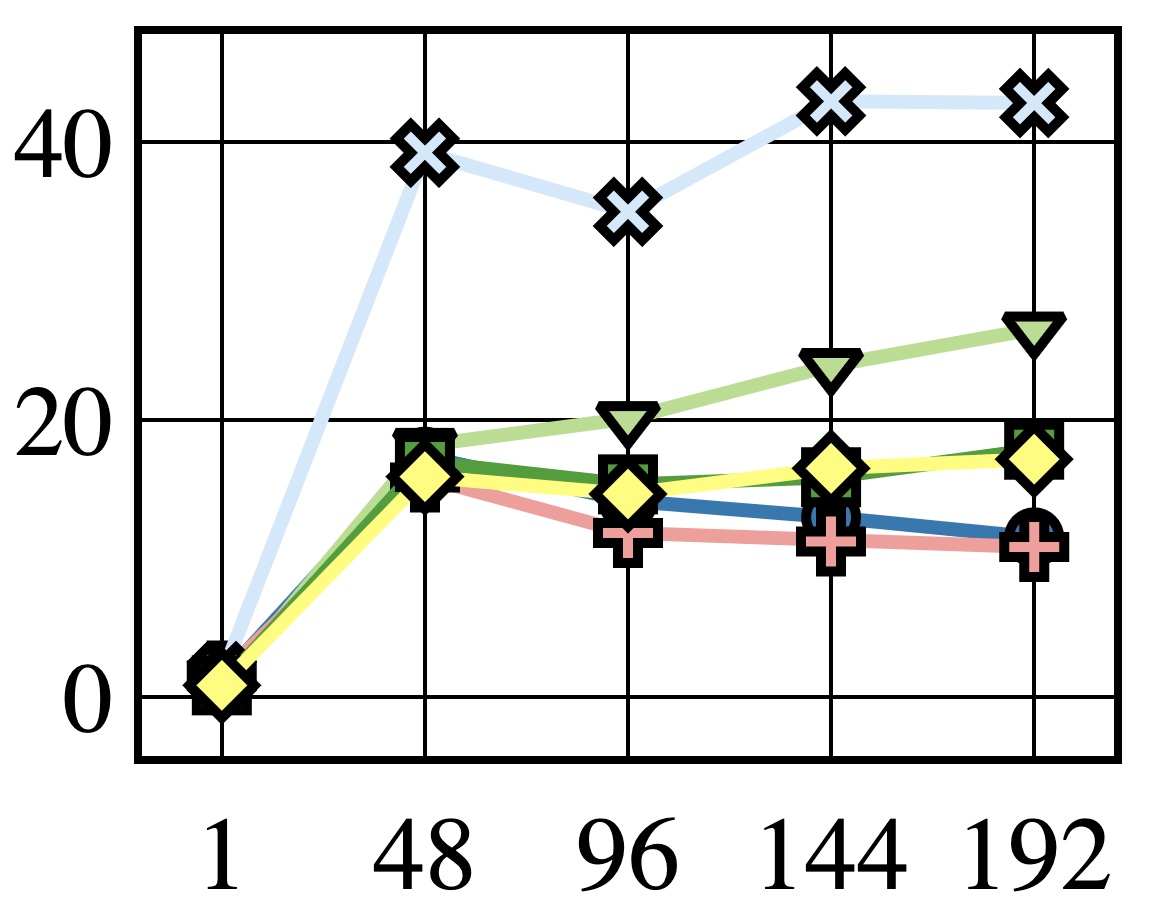}    
        \caption{$90-0-10$}
        \label{fig2:ct-90-0-10}
    \end{subfigure}
    \hspace{-7pt}
    \begin{subfigure}{0.17\textwidth}
        \includegraphics[width=\textwidth]{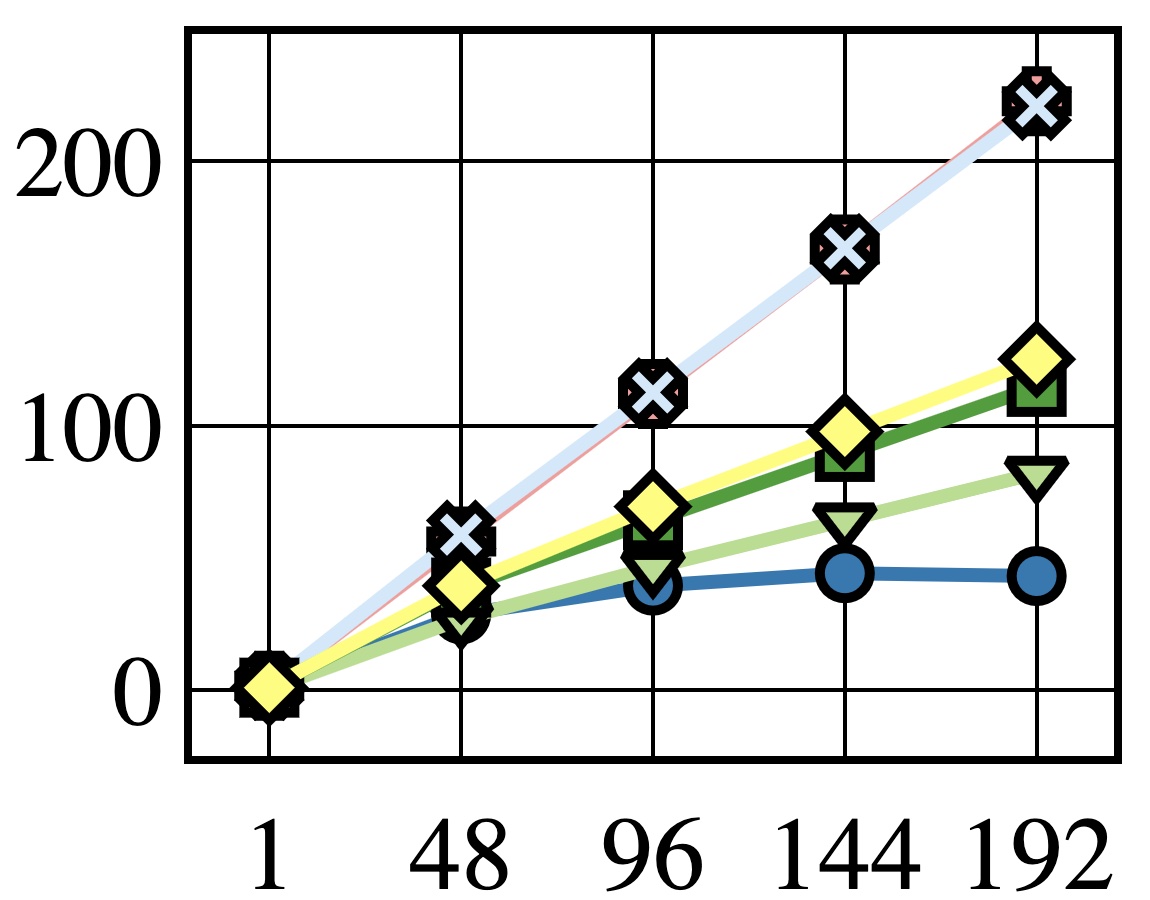}
        \caption{$0-90-10$}
        \label{fig2:ct-0-90-10}
    \end{subfigure}
    \hspace{-7pt}
    \begin{subfigure}{0.17\textwidth}
        \includegraphics[width=\textwidth]{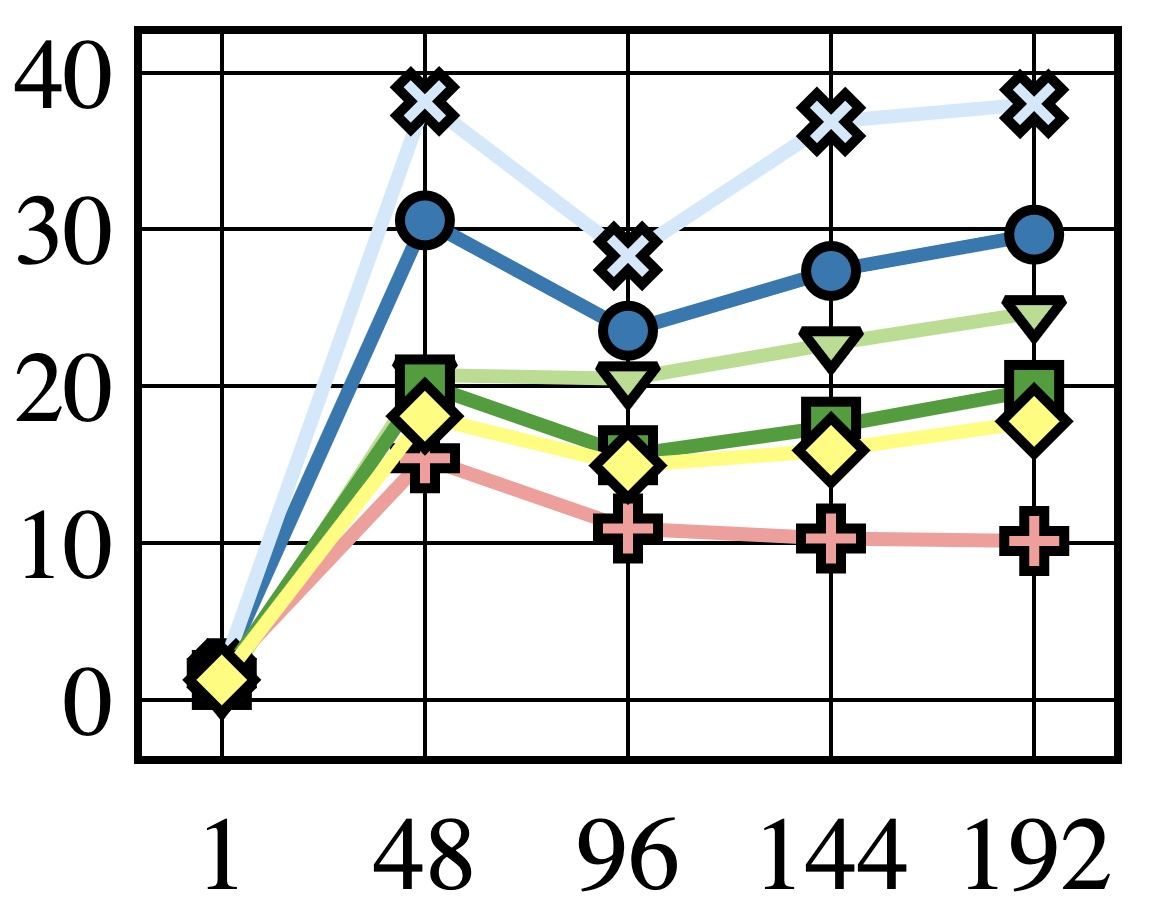}
        \caption{$100-0-0$}
        \label{fig2:ct-100-0-0}
    \end{subfigure}\\
    \caption{Throughput (Mops/s) under various workload configurations for the skip list (SL) (a) -- (f) and Citrus tree (CT) (g) -- (l), \cradd{while varying} the number of threads on the x-axis. Workloads are written as $U-C-RQ$, corresponding to the percentages of update ($U$), contains ($C$) and range queries ($RQ$).}
    \label{fig:fig2}
\end{figure*}

\underline{Traversal Instrumentation.}
Traversal-dominant data structures (e.g., linear data structures like a linked list) offer insight into the overhead required by each algorithm for supporting range queries.
The most obvious consequence of the above observation is that the performance of strategies that require every access to be instrumented suffer compared to those that support uninstrumented traversals.
Figure~\ref{fig:lazylist} shows the total throughput of operations on a lazy list as a function of the number of worker threads at both 10\% updates and 90\% updates with 10\% range queries.

At 10\% updates, Bundle avoids paying a per-node cost for all operations by performing at least a portion of traversals using regular links.
Only after landing in the (proximity of the) target range (or key) Bundle starts incurring an overhead.
Although internally EBR-RQ uses double-compare single-swap (DCSS), which may add overhead per-dereference, its design is such that a single pointer can be used to reference both the descriptor and the next node.
RLU's overhead is concentrated in updates because they wait for ongoing reads; this cost becomes apparent at higher thread-counts.
In contrast to the other competitors, vCAS requires an additional dereference per-node to maintain linearizability, leading to less than half the throughput of bundling.

At 90\% update rates, the drawbacks of other competitors start to become visible.
For example, RLU's heavy update synchronization causes its performance to fall below that of vCAS.
Bundle and EBR-RQ maintain higher performance because of their low instrumentation costs.


\underline{Varying Workload Mix.}
Figure~\ref{fig:fig2} reports the throughput (Mops/s) of different workload mixes as a function of thread count for both the skip list (Figures~\ref{fig2:sl-2-88-10}-\ref{fig2:sl-100-0-0}) and Citrus tree (Figures~\ref{fig2:ct-2-88-10}-\ref{fig2:ct-100-0-0}).
Except for the two rightmost columns, the range query percentage is 10\%, while the update and contains percentages vary. 
\crremove{For space constraints, results for 2\% and 50\% range query operations can be found in Section~\ref{appendix:addition-plots} of the Supplemental Materials, and the trends are similar.}
\cradd{Results for 2\% and 50\% range query operations can be found in the companion technical report~\cite{nelson2021bundling}.}

Our first general observation is that our bundled data structures perform best compared to RLU and EBR-RQ when the workload is mixed.
In the four left-most columns, Bundle matches or surpasses the performance of the next best strategy, in most cases.
Under these mixed workloads, Bundle achieves maximum speedups compared to EBR-RQ of 1.8x in the skip list (Figure~\ref{fig2:sl-2-88-10}) and 3.7x in the Citrus tree (Figure~\ref{fig2:ct-2-88-10}).
Against RLU, bundling can reach 1.8x better performance (Figure~\ref{fig2:ct-50-40-10}).
For both, the best speedups are achieved at 192 threads.
Overall, the mixed configurations represent a wide class of workloads; whereas, the two rightmost columns represent corner case workloads and are discussed separately, later.
Although we do not test an RLU implementation of the skip list, based on our observations of the Citrus tree results, we expect it would offer similar performance to Bundle for lower update percentages.




In more detail, Bundle performs better than EBR-RQ in low-update workloads by avoiding high-cost range queries.
Bundling localizes overhead by waiting for pending entries, limiting the impact of concurrent updates on range query performance.
In contrast, EBR-RQ's design allows that possibly many additional nodes be visited in the limbo lists to finalize a snapshot.
For the skip list, even at 2\% updates, we measure EBR-RQ traversing an average of 141 additional nodes \emph{per-range query} at 48 threads and up to 575 nodes at 192 threads.
The above numbers are nearly identical for both the 
lazy list and Citrus tree, indicating that this behavior is independent of data structure.

The case is similar for the Citrus tree, where Bundle outperforms EBR-RQ by an average 1.2x across mixed workloads at 48 threads, except 90\% where it is 10\% worse.
At higher thread counts, bundling has the clear advantage.
For all workloads at 48 threads, Bundle does not perform better than RLU.
At a high number of threads, Bundle is never worse.
For example, when more updates are present in the workload (Figures~\ref{fig2:ct-50-40-10} and~\ref{fig2:ct-90-0-10}), Bundle outperforms RLU (1.8x and 1.6x) at 192 threads.
Interestingly, at high update workloads the Citrus tree does not scale well beyond one NUMA zone, regardless of the presence of range queries.

Compared to vCAS, Bundle performs up to 1.2x-1.5x better in low update workloads (i.e., 0\%, 2\% and 10\% updates) across all thread configurations due to its lightweight traversal.
In update-intensive workloads (i.e., 50\% and 90\% updates), Bundle and vCAS have very similar performance for the skip list, within 5\%.
This is because their update operations have similar execution patterns (i.e., adding a new version to a lists) and the underlying data structure synchronization starts impacting performance.
In the Citrus tree, vCAS demonstrates better performance at 50\% and 90\% updates and high thread count: as much as 1.5x.
This is because operations in vCAS are able to help other concurrent operations, while Bundle must wait for them to finish.
Exploring whether bundling can exploit a similar behavior is an interesting future work.


To better understand the (extreme) cases in which Bundle performance is surpassed by either RLU or EBR-RQ, we note that RLU and EBR-RQ prefer workloads at opposing ends of the configurations spectrum.
In the read-only setting (Figure~\ref{fig2:ct-0-90-10}), RLU performs well because there are no synchronization costs, and achieves performance nearly equivalent to Unsafe in the Citrus tree.
\cradd{We also evaluate the case of range query only workload and we obserse identical trends as the read-only case (Figures~\ref{fig2:sl-100-0-0} and~\ref{fig2:ct-100-0-0}).}

\cradd{Despite its impressive performance} \crremove{However}, recall that even a low percentage of updates is enough to cause \cradd{RLU} \crremove{this impressive performance} to collapse (see Figure~\ref{fig2:ct-2-88-10}), namely from the synchronization enforced by writers (i.e., RLU-sync).
Hence, when a workload is primarily updates (Figures~\ref{fig2:ct-50-40-10}, \ref{fig2:ct-90-0-10}, and \ref{fig2:ct-100-0-0}), RLU is worst.
On the other hand, EBR-RQ range queries increment a global epoch counter and validate their snapshot, which leads to worse performance in read-dominated workloads compared to update-intensive ones.
In fact, at update-only workloads (the rightmost column), EBR-RQ behaves nearly as well as Unsafe, while for read-only (the second-to-last column) workloads it is the worst competitor.

The performance of both Bundle and vCAS in those extreme cases demonstrates that they properly handle this trade-off, with Bundle outperforming in the read-only cases due to its lightweight traversal.

\underline{Varying Range Query Size.}
~Figure~\ref{fig:rqlen} shows the throughput (Mops/s) of update and range query operations on a Citrus tree when a single NUMA zone is partitioned such that half of the hyperthreaded cores (i.e., 24) are running 100\% updates and the other half are running 100\% range queries, \cradd{while varying the range query length.}
\crremove{We also increase the range query length to understand how the two operations interact.}

Our first observation is that for Bundle, EBR-RQ, and vCAS, the performance of update operations is not negatively impacted by the length of range queries, while for RLU this is not the case.
In the former three strategies, update operations are not blocked by range queries, whereas in the latter they must synchronize with ongoing read operations.
The result is that for RLU, longer running range queries have a detrimental impact on update performance.

Considering range query performance, EBR-RQ throughput is effectively flat across all range query sizes.
This is attributed to the required checking of limbo-lists and the helping mechanism to support DCSS.
The result is that it dominates the cost of range queries and even short ones perform poorly.
On the other hand, Bundle, vCAS, and RLU demonstrate much better performance for short range query operations (i.e., >10x over EBR-RQ). 
The disparity is primarily due to the cost of EBR-RQ range queries depending on the number of concurrent update operations, while Bundle, vCAS, and RLU only incur local synchronization costs.

\subsection{Database Integration Performance}

We evaluate the application level performance of bundling by running the TPC-C benchmark in DBx1000~\cite{stonebreaker}, an in-memory database system, integrated with our bundled skip list and Citrus tree. 
The data structures implement the database indexes and the number of warehouses is equal to the number of threads.
Threads access other warehouses with a 10\% probability.

\begin{figure}[t]
    \centering
        \includegraphics[width=0.47\textwidth]{figures/legend-microbench.jpg}\\
    \hspace{-25pt}
    \begin{subfigure}{.04\textwidth}
        \centering
        \includegraphics[width=.75\textwidth]{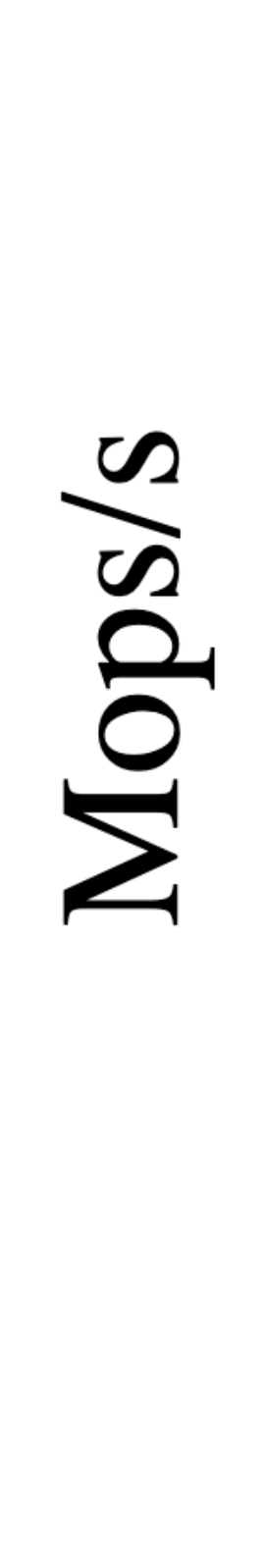}
    \end{subfigure}
    \begin{subfigure}{.4\textwidth}
        \centering
    \includegraphics[width=1.1\textwidth]{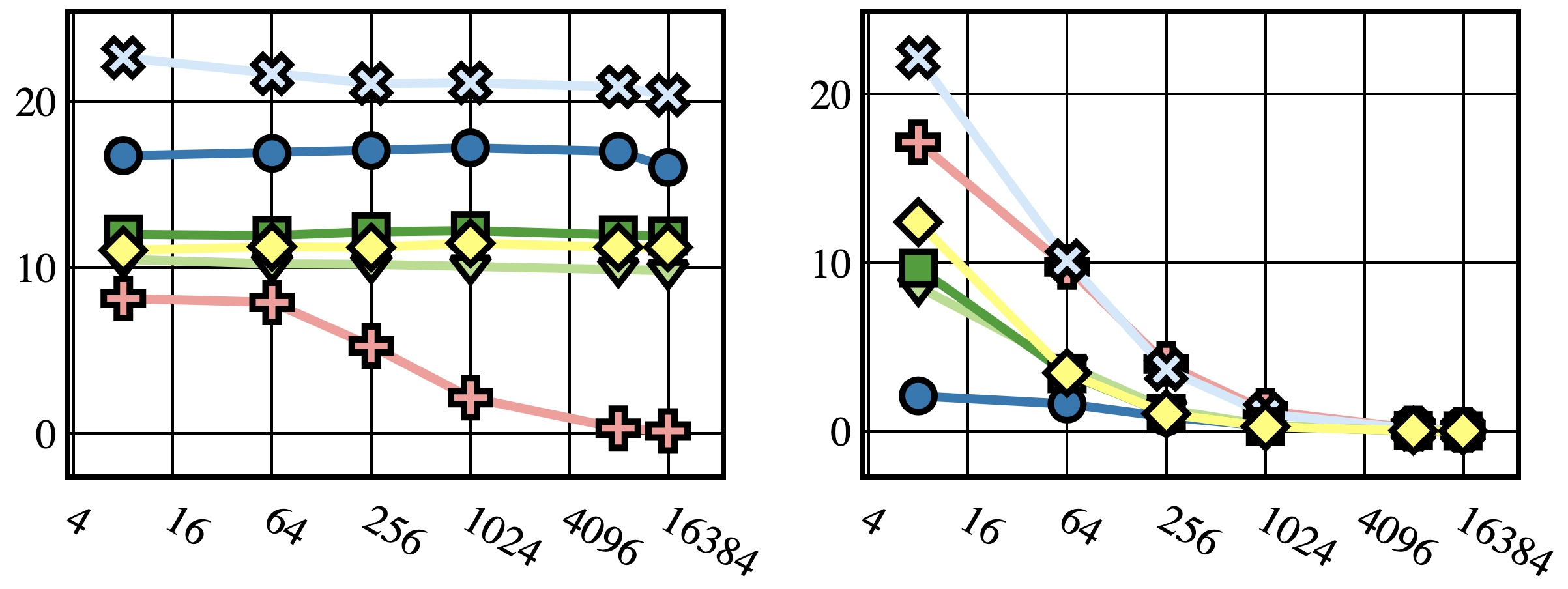}
    \end{subfigure}
    \vspace{-15pt}
    \caption{Update (left) and range query (right) throughput for 24 update-only threads and 24 range-query-only threads \cradd{while varying the range query size.}}
    \label{fig:rqlen}
\end{figure}

\begin{figure}[t]
    \centering
        \includegraphics[width=0.38\textwidth]{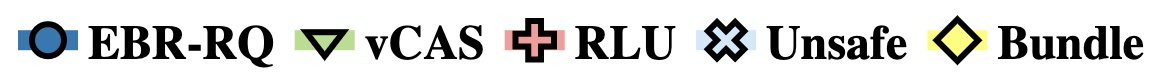}\\
    \vspace{-5pt}
    \hspace{-25pt}
    \begin{subfigure}{.04\textwidth}
        \centering
        \includegraphics[width=.75\textwidth]{figures/fig2yaxis.pdf}
    \end{subfigure}
    \hspace{-8pt}
    \begin{subfigure}{.4\linewidth}
        \includegraphics[width=1.2\textwidth]{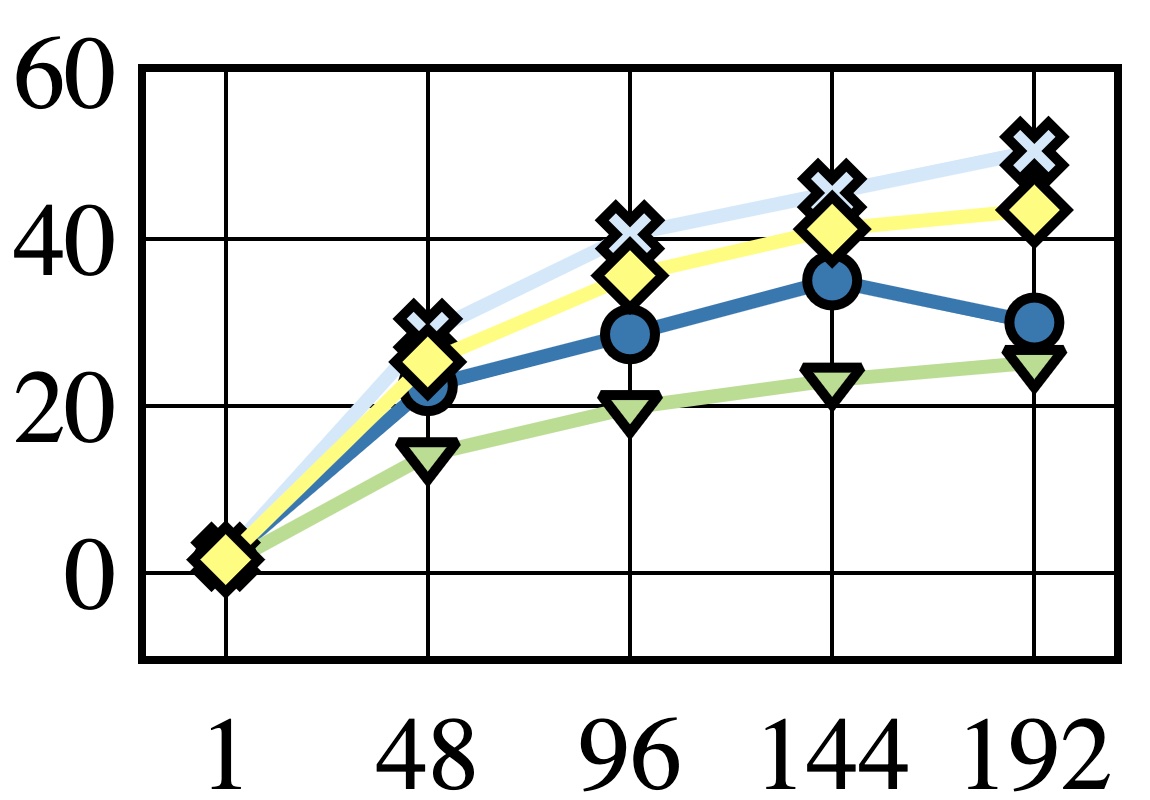}
    \end{subfigure}
    \hspace{15pt}
    \begin{subfigure}{.4\linewidth}
        \centering
        \includegraphics[width=1.2\textwidth]{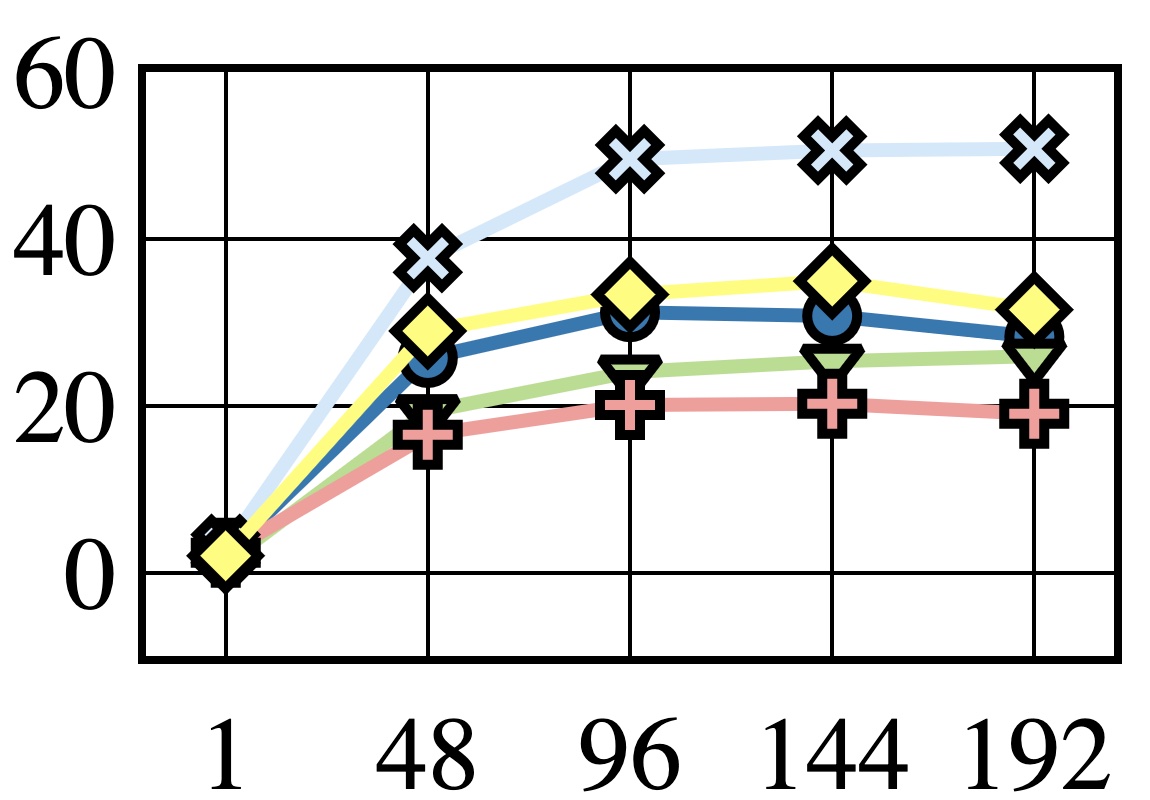}
    \end{subfigure}
    \vspace{-10pt}
    \caption{Throughput (Mops/s) of index operations in DBx1000 running the TPC-C benchmark \cradd{at different thread counts}. Skip list (left) and Citrus tree (right). There is no RLU-based skip list.}
    \label{fig:macrobench}
\end{figure}

We use the \texttt{NEW\_ORDER} (50\%), \texttt{PAYMENT} (45\%) and \texttt{DELIVERY} (5\%) transaction profiles.
The \texttt{DELIVERY} profile is particularly interesting since its logic includes a range query over the index representing the new order table, ordered by \texttt{order\_id}, with the goal of selecting the oldest order to be delivered. Next, the order is deleted to prevent subsequent \texttt{DELIVERY} transactions from delivered the same order again. 
In our experiments, the range query selects the oldest order in the last 100 orders.
A \texttt{PAYMENT} transaction performs a range query on the customer index to look up a customer by name with 60\% probability and then updates the total income according to the amount on the payment.
\texttt{NEW\_ORDER} modifies multiple tables and updates their indexes accordingly, including the new order index.

We report the total index throughput measured across the system when the skip list and Citrus tree is used as an index (Figure~\ref{fig:macrobench}).
Note that we elide a comparison against the baseline DBx1000 index since it is a hashmap and does not support range queries.
Our results demonstrate that bundling outperforms all competitors regardless of the number of threads used. 
Summarizing our findings, 
Bundle is on average, across all thread configurations, only 13\% worse than Unsafe for the skip list but 29.6\% for the Citrus tree.
In comparison to the next best competitor (EBR-RQ), Bundle achieves 1.4x and 1.1x better performance for the skip list and Citrus tree, respectively.



\toremove{
\subsection{Discussion}
Bundling manages the trade-off between update-intensive and read-only workloads effectively.
The overhead of updating bundles is relatively low, which improves upon RLU's synchronization costs. 
On the other hand, limiting range query traversals to the necessary nodes improves upon EBR-RQ.
Compared to vCAS, bundling shines in traversal dominant data structures because of its optimistic \texttt{pre-range} phase.
Hence, performance consistency across different workloads is an important byproduct of bundling.
In update-heavy workloads, the dominating contributor to performance is the throughput of update operations themselves.

Unlike its competitors, bundling does not sacrifice performance in one case for higher performance in the other, which leads to overall better throughput.
Our evaluation demonstrates that our
design is better suited for systems that have different workloads on different internal data structures, as is typical in database systems, without the need for multiple implementations each targeting specific workload distributions.
}
\section{Conclusion}
\label{sec:conclusion}

We presented three concurrent linked data structure implementations 
deploying a novel building block, called bundled references, to enable range linearizable query support. Bundling data structure links with our building block shows that the coexistence of range query and update operations does not forgo achieving high-performance, even in lock-based data structures.

\section*{Acknowledgments}
Authors would like to thank the shepherd and the anonymous reviewers for all the constructive comments that increased the quality of the final version of the paper. This material is based upon work supported by the National Science Foundation under Grant No. CNS-1814974 and CNS-2045976.


\bibliography{references}

\newpage
\appendix
\section{Artifact Description}

Here we give a high-level overview of the artifact supporting this paper. 
For more detailed instructions on how to configure, run and plot the results, we refer the reader to the comprehensive README file included in the root directory of the source code (\url{https://zenodo.org/record/5790039}).
Our implementation graciously builds upon the existing benchmark found in~\cite{ebr-rq}.

\subsection{Requirements} 
The project is written primarily in C++ and was compiled with GNU Make using the GCC compiler with the compilation flags \texttt{-std=c++11 -mcx16 -O3}.
The two primary dependencies for the project are the \texttt{jemalloc} memory allocation library and the \texttt{libnuma} NUMA policy library.
The plotting scripts are written in Python and rely on a handful of libraries.
For example, we use the Plotly plotting library to generate interactive plots.
A full list of dependencies can be found in the README.
Note that NUMA must be enabled to run properly and it is recommended to run on a machine with at least 48 threads to reproduce our results.

\textbf{Docker.}
For convenience, we also include a Dockerfile that handles setting up the environment.
However, if using Docker, keep in mind that performance may suffer and the results may differ from those presented in the paper.

\subsection{Code Organization}
The root directory contains implementations of all range query techniques, data structures incorporating them, and the benchmarks used for our evaluation.
The following is a list of notable directories that pertain to our design and evaluation.

\vspace{10pt}
\textbf{Relevant directories:}
\begin{itemize}
 \item \texttt{bundle} (bundled reference implementation)
 \item \texttt{bundle\_lazylist} (bundled lazylist implementation)
 \item \texttt{bundle\_skiplist\_lock} (bundled skip list implementation)
 \item \texttt{bundle\_citrus} (bundled Citrus tree implementation)
 \item \texttt{vcas\_lazylist} (vCAS-based lazylist)
 \item \texttt{vcas\_skiplist\_lock} (vCAS-based skip list)
 \item \texttt{vcas\_citrus} (vCAS-based Citrus tree)
 \item \texttt{rq} (range query provider implementations)
 \item \texttt{microbench} (microbenchmark code)
 \item \texttt{macrobench} (macrobenchmark code)
\end{itemize}
\vspace{10pt}

Our primary contributions are located in the directories prefixed with \texttt{bundle}, which implement the bundled references used when bundling data structures, as well as our bundled versions of the lazylist, skip list and Citrus tree.
The same naming strategy is used for the RLU- and vCAS-based approaches.
The \texttt{rq} directory contains implementations of range query providers, including the one used by bundling, which defines the interaction with the global timestamp.

The \texttt{microbench} and \texttt{macrobench} directories contain the code used for our evaluation.
The microbenchmark (i.e., \texttt{microbench}) consists of two experiments designed to understand the implications of design choices between competitors, corresponding to Figure~\ref{fig:lazylist}, Figure~\ref{fig:fig2} and Figure~\ref{fig:rqlen} in the paper.
The first experiment tests various update rates while fixing the range query size to 50 and the range query rate to 10\%. 
Results of this experiment are shown in the first two of the aforementioned figures.
The second experiment measures performance of a workload with 50\% updates and 50\% range queries, while adjusting the range query size, and corresponds to Figure~\ref{fig:rqlen}.
This experiment uses dedicated threads to perform each operation type to highlight the effects of updates on range queries, and vice versa.
Finally, the macrobenchmark (i.e., \texttt{macrobench}) integrates the data structures as indexes in the DBx1000 in-memory database.
For both, data is saved to the benchmark's \texttt{data} sub-directory, which is then processed by the respective \texttt{make\_csv.sh} scripts and plotted using \texttt{plot.py} (located in the root directory).

\subsection{Initial Configuration}
Configuring the above benchmarks relies on six parameters that must be adjusted according to the testbed, prior to compilation. The following can be found in the \texttt{config.mk} file located in the root directory:

\begin{itemize}
 \item \texttt{allocator}, the allocator used by the benchmarks (blank represents the system allocator)
 \item \texttt{maxthreads}, the max number of concurrent threads
 \item \texttt{maxthreads\_powerof2}, \texttt{maxthreads} rounded up to a power of 2
 \item \texttt{threadincrement}, the thread step size
 \item \texttt{cpu\_freq\_ghz}, processor frequency for timing
 \item \texttt{pinning\_policy}, the thread affinity mapping
\end{itemize}

Note that the maximum number of threads cannot exceed the maximum number of threads in the testbed.
See the Configuration Tips section of the README for additional information and examples.

\subsection{Microbenchmark}

For the following section we assume that the current working directory is \texttt{microbench}.
To build the benchmark, use the command: \texttt{make -j lazylist skiplistlock citrus rlu}.
Although each data structure binary can be run independently to test specific configurations (see the README), there are scripts to generate a full suite of configurations automatically.
Once compiled, the experimental configurations presented in this paper can be run using \texttt{./runscript.sh}.

Internally it calls \texttt{./experiment\_list\_generate.sh}, which configures the range query techniques, data structures, and sizes to test and defines the experiments shown in the paper.
The \texttt{./runscript.sh} script then reads the set of experimental configurations generated and starts executing them.
The number of trials and length of each can be adjusted in the \texttt{./runscript.sh} file.
The last important configuration file in the microbenchmark is \texttt{supported.inc}, which determines which combination of data structure and size configurations are tested.

A file called \texttt{summary.txt} is created in the \texttt{./data} that contains an overview for the entire run.
Individual results are stored in sub-directories whose path is\\ \texttt{./data/<experiment>/<data\_structure>}.
If there are any errors or warnings during execution they will be reported in the \texttt{./warnings.txt} file.
We describe how to generate plots from the results in Section~\ref{sec:plotting}.

\subsection{Macrobenchmark}

The macrobenchmark demonstrates the performance of competitors by replacing the index of the in-memory database system called DBx1000~\cite{stonebreaker}.
We assume for the remainder of this section that the working directory is \texttt{./macrobench}.

To build this portion of the project, run the \texttt{./compile.sh} script.
Once complete, the macrobenchmark can then be executed with the \texttt{./runscript.sh} script.
This script is pre-configured to execute the TPC-C benchmark using a workload mix consisting of 50\% \texttt{NEW\_ORDER}, 45\% \texttt{PAYMENT} and 5\% \texttt{DELIVERY} transactions.
Similar to the microbench, all results are saved in the \texttt{./data} sub-directory.

\subsection{Plotting Results}
\label{sec:plotting}
In order to generate plots, we first translate the results produced by the respective \texttt{./runscript.sh} scripts into CSV files for ingestion by our plotting script.
This is handled automatically by the \texttt{plot.py} Python script, which calls the \texttt{make\_csv.sh} script within each benchmark directory.

After the microbenchmark and macrobenchmark have been executed successfully, plots can be generated using \texttt{python plot.py $--$save\_plots} and including either the \texttt{$--$microbench} or \texttt{$--$macrobench} command line flag.
This will first generate the requisite CSV file in the corresponding data sub-directory and then create plots from the data contained in the CSV files.
The experimental configuration is automatically loaded, but we note that other command line flags can override their values, if necessary.
Figures are saved as interactive HTML files in the \texttt{./figures} directory, located in the root directory of the project.
The plots can be viewed using a browser.

\subsection{Expected Output Files}
Compiling, running and generating plots produces a number of output files. Here, we provide a summary of the expected output after all of the above steps have been executed. Note that we use angle brackets to indicate that the contents are repeated for the range query techniques (i.e., \texttt{rq}) and data structures (i.e., \texttt{ds}).
Although other files are generated by the compilation phases (e.g., object files), we do not list them here but instead focus on those of most importance to the user.
We assume that the current directory is the project's root directory.

\DTsetlength{0.2em}{1em}{0.2em}{0.4pt}{1.6pt}
{\Small
\dirtree{%
.1 ./.
.2 microbench/.
.3 data/.
.4 summary.txt.
.4 rq\_sizes/.
.5 <rq>/.
.6 *.out.
.4 workloads/.
.5 <rq>/.
.6 *.out.
.3 <machine\_name>.<ds>.<rq>.out.
.2 macrobench/.
.3 data/rq\_tpcc/.
.4 summary.txt.
.4 *.out.txt.
.3 bin/<machine\_name>/rundb\_TPCC\_IDX\_<ds>\_<rq>.out.
.2 figures/.
.3 microbench/.
.4 rq\_sizes/.
.5 <ds>/.
.6 *.html.
.4 workloads/.
.5 <ds>/.
.6 *.html.
.3 macrobench/.
.4 <ds>/.
.5 *.html.
}
}

\section{Correctness}
\label{appendix:correctness}

We prove that range query operations are linearizable by showing that \textit{i)} a range query visits all nodes that belong to its linearizable snapshot, and that \textit{ii)} a range query does not visit any other node that does not belong to its linearizable snapshot. By extension, we also prove the correctness of contains operations, which are treated as a range query of one element. Finally, we prove that optimizing contains to elide reading the global timestamp does not compromise its correctness. Proving the linearization of the other operations (i.e., insert and remove) follows from the correctness arguments of the original data structure.

\begin{theorem}
\label{theorem:non-restartable}
Bundled data structures are linearizable.
\end{theorem}

\begin{proof}
The linearization point of updates is Line 9 of Algorithm~\ref{algo:prepare} (i.e., incrementing \texttt{globalTS}). That is because the global timestamp is only incremented after the necessary pending bundle entries are installed, blocking conflicting update operations. Hence, the global timestamp enforces a total order among conflicting update operations.

The next step is to prove that range queries are linearized at the moment they read \texttt{globalTS} (Line 9 in Algorithm~\ref{algo:rangequery}). Consider a range query operation RQ whose linearization point is at timestamp $t_{RQ}$.
Once RQ reaches this linearization point, RQ holds a reference to a node that has been reached using standard links in the \texttt{pre-range} phase (we call it $\mathcal{N}$). We first show that $\mathcal{N}$ has three important characteristics:
\begin{itemize}
    \item RQ is guaranteed to find a bundle entry in $\mathcal{N}$ that satisfies $t_{RQ}$. This can be shown as follows. 
    Since updates are assumed to only make a node reachable after incrementing the global timestamp, $\mathcal{N}$ either has a pending bundle entry or a finalized one. In either case, RQ necessarily arrives after the increment and must read a timestamp greater than or equal to the timestamp obtained by the update that inserted $\mathcal{N}$.
    
    \item $\mathcal{N}$ can be used to sequentially (i.e., in absence of concurrent updates) traverse all the nodes in the range using \texttt{getNext}.
    This property should be guaranteed by the data structure specific implementation of its \texttt{getNext} function. For example, the \texttt{pre-range} phase in the linked list ends at a node whose value is less than the lowest value in the range.
    
    \item After reading the timestamp (Line~\ref{line:readglobalts} in Algorithm~\ref{algo:rangequery}), $\mathcal{N}$ can be used to sequentially (i.e., in absence of concurrent updates) traverse to the first node in the range using \texttt{getNextFromBundle}. 
    This property should also be guaranteed by the data structure specific implementation. For example, in the linked list we have two cases. The first case is when $\mathcal{N}$ is still in the list after reading the timestamp, which means that trivially a sequence of \texttt{getNextFromBundle} will lead to the first node in the range (just like \texttt{getNext}). The second case is when $\mathcal{N}$ is deleted after the \texttt{pre-range} phase and before reading the timestamp.
    In this case, adding a bundle entry in $\mathcal{N}$ during its removal redirects the range query to the head node and enforces a consistent traversal to the first node of the range using \texttt{getNextFromBundle}.
\end{itemize}

By leveraging these three characteristics, RQ can sequentially (i.e., in absence of concurrent updates) traverse all nodes in the range using \texttt{getNextFromBundle}.

It remains to demonstrate how the \texttt{enter-range} and \newline\texttt{collect-range} phases (i.e., after reading the timestamp) interact with concurrent updates. To cover those cases, assume an arbitrary key $k$, within RQ's range. We then have the following possible executions:  
\begin{enumerate}
    \item No node with a key equal to $k$ exists in the data structure when the range query is invoked, and either:
    \begin{enumerate}
        \item No insert operation on $k$ is executed concurrently with RQ. Thus, RQ traverses using \texttt{GetNextFromBundle} and will not visit any node with key $k$.
        \item An insert($k$) operation INS is executed concurrently with RQ, for which there are two possibilities:
        \begin{enumerate}
            \item INS is linearized before RQ. If INS increments \texttt{globalTS} before RQ takes its snapshot of \texttt{globalTS}, it implies that INS also installed the necessary pending bundle entries before RQ reads \texttt{globalTS}. This means that RQ will block until INS finishes, and thus RQ will include $k$ since INS finalizes the bundles with a timestamp satisfying $t_{RQ}$.
            \item INS is linearized after RQ. If INS increments \texttt{globalTS} after RQ takes its snapshot of \texttt{globalTS}, then regardless of how RQ proceeds it will not include $k$ because the bundle entries finalized by INS do not satisfy $t_{RQ}$, and thus will not be traversed using \texttt{GetNextFromBundle}.
        \end{enumerate}
    \end{enumerate}
    \item There is a node with key equal to $k$ in the data structure when the range query is invoked, and either:
    \begin{enumerate}
        \item No remove operation on $k$ is executed concurrently with RQ. In this case, RQ will include $k$ in its snapshot using \texttt{GetNextFromBundle} to collect the result set.
        \item A remove($k$) operation REM is executed concurrently with RQ. Again, two situations arise:
        \begin{enumerate}
            \item REM is linearized before RQ. If REM increments \texttt{globalTS} before RQ takes the snapshot of \texttt{globalTS}, this implies that REM also installs the necessary pending Bundle entries before RQ takes the snapshot, which means that RQ will have to block until REM finishes, and thus RQ will not include $k$.
            \item REM is linearized after RQ. If REM increments \texttt{globalTS} after RQ takes the snapshot of \texttt{globalTS}, then regardless of how RQ proceeds it will include $k$ because the bundle entries finalized by REM have a larger timestamp, and thus $k$ is still reachable using \texttt{GetNextFromBundle}.
        \end{enumerate}
    \end{enumerate}
    \item Any more complex execution that includes multiple concurrent update operations on $k$ can be reduced to one of the above executions since those update operations are conflicting and will block each other using the pending bundle entries.
\end{enumerate}
Finally we show that contains operation is linearizable. Unfortunately, we cannot use the same linearization point of range query for contains because of our optimization that eliminates the step of reading \texttt{globalTs}. However, it is still easy to reason about the linearizability of contains by inspecting the moment at which it dereferences the last bundle entry in the \texttt{enter-range} phase. We have two scenarios. If the update operation that adds this entry increments \texttt{globalTs} before the contains starts, then the contains can be linearized at the moment it starts. Otherwise, the contains is linearized right after the \texttt{globalTs} increment occurs, which is necessarily between the invocation and the return of the contains operation since the contains will block on the entry if it is not yet finalized.
\end{proof}

\section{Additional Plots}
\label{appendix:addition-plots}

Figure~\ref{fig:extra-plots} summarizes our performance evaluation at 2\% and 50\% range query operations. Foremost, in both configurations bundling continues to perform well under read-dominant workloads. At higher update workloads, EBR-RQ performs best in the 2\% range query workloads, because of its limited contention on the global timestamp for this workload. Meanwhile bundling and RLU perform similarly to results show in Figure~\ref{fig:fig2}. In contrast, 50\% range query workloads tend to favor RLU over EBR-RQ, with bundling equaling or outperforming RLU.

\begin{figure}[h]
    \centering
    \begin{subfigure}{\linewidth}
        \centering
        \includegraphics[width=.8\textwidth]{figures/legend-microbench.jpg}
    \end{subfigure}
    \begin{subfigure}{0.17\textwidth}
        \includegraphics[width=\textwidth]{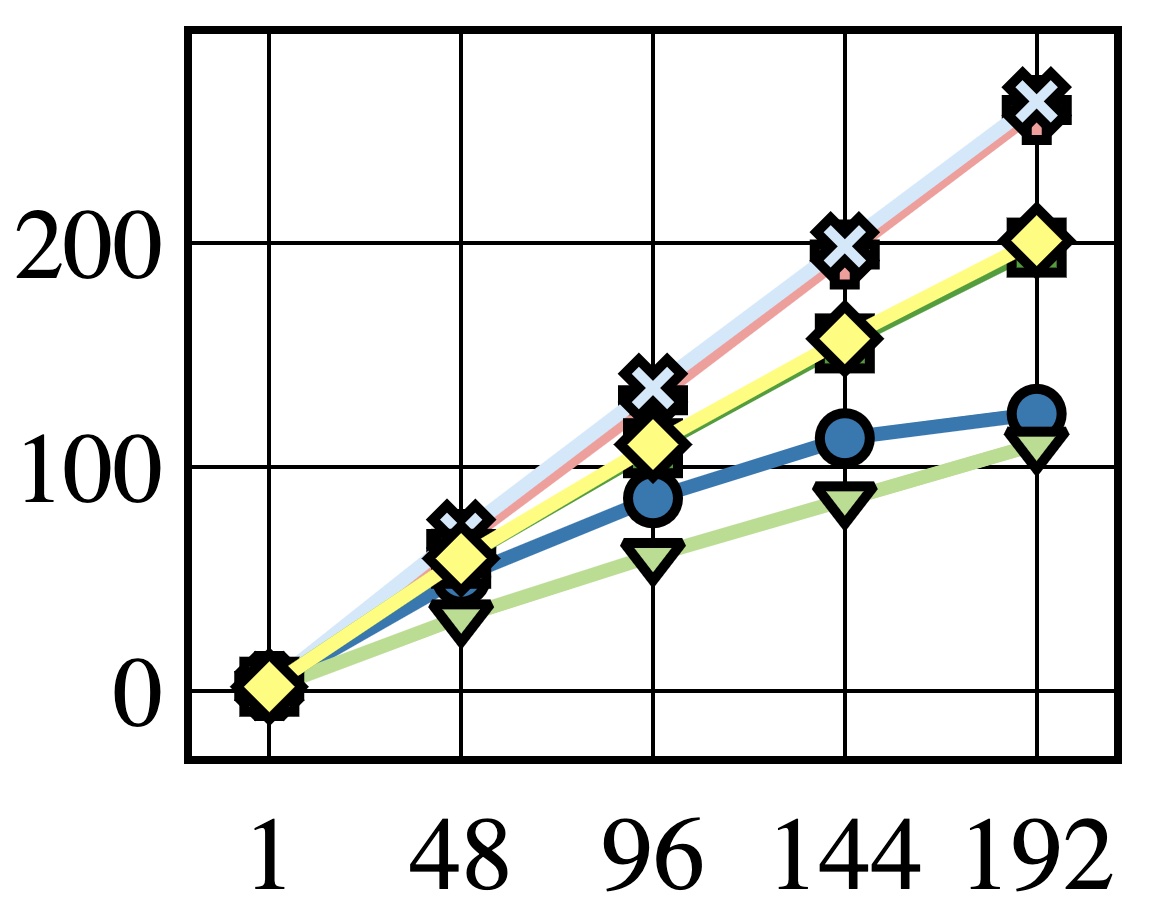}
        \caption{SL, $0-98-2$}
    \end{subfigure}
    \begin{subfigure}{0.17\textwidth}
        \includegraphics[width=\textwidth]{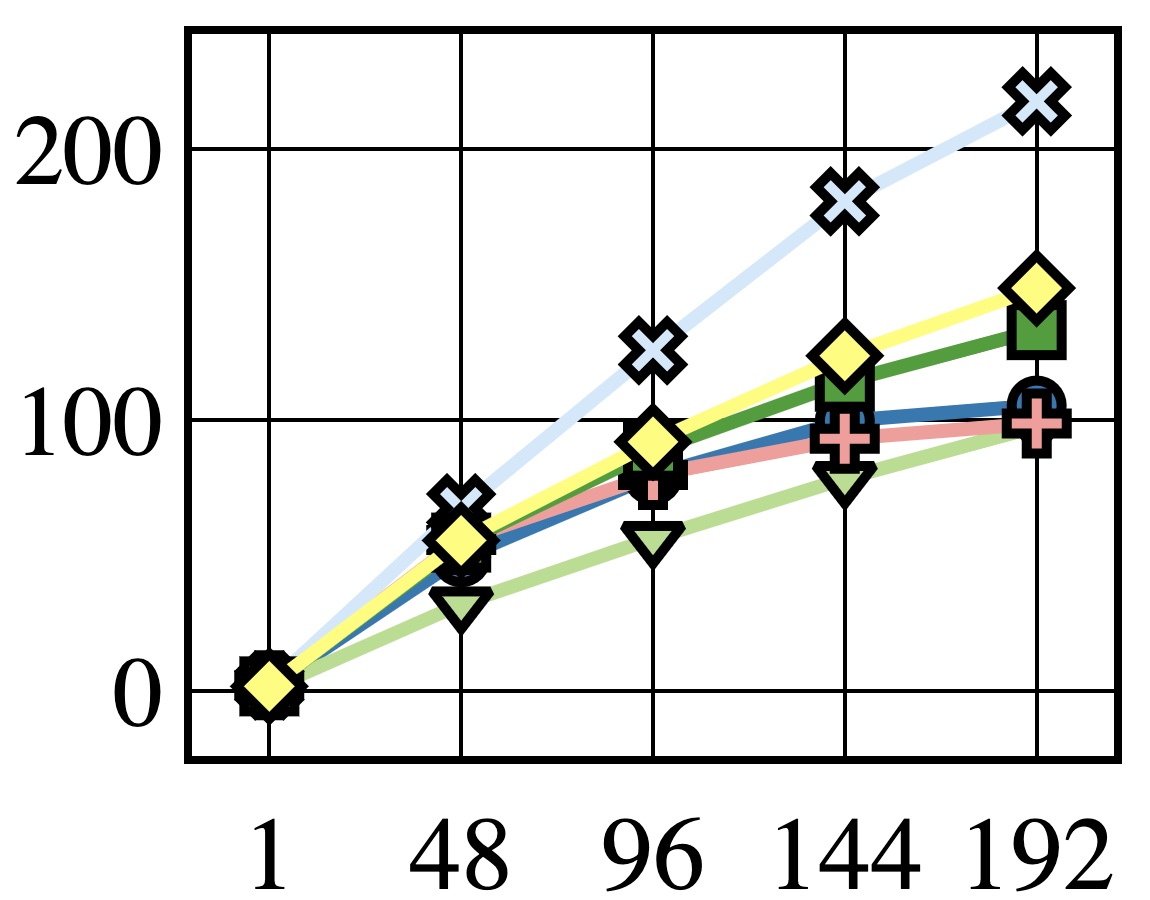}
        \caption{SL, $2-94-2$}
    \end{subfigure}
    \begin{subfigure}{0.17\textwidth}
        \includegraphics[width=\textwidth]{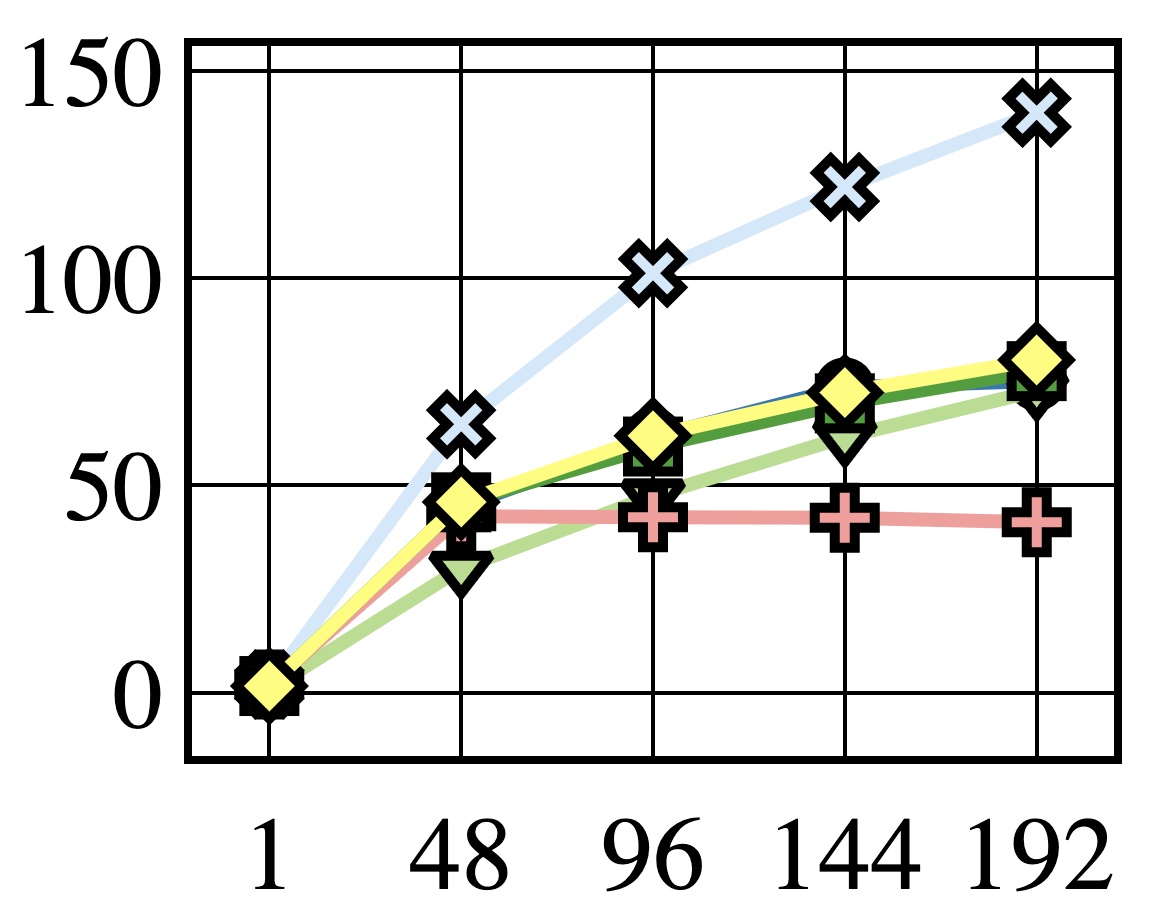}
        \caption{SL, $10-88-2$}
    \end{subfigure}
    \begin{subfigure}{0.17\textwidth}
        \includegraphics[width=\textwidth]{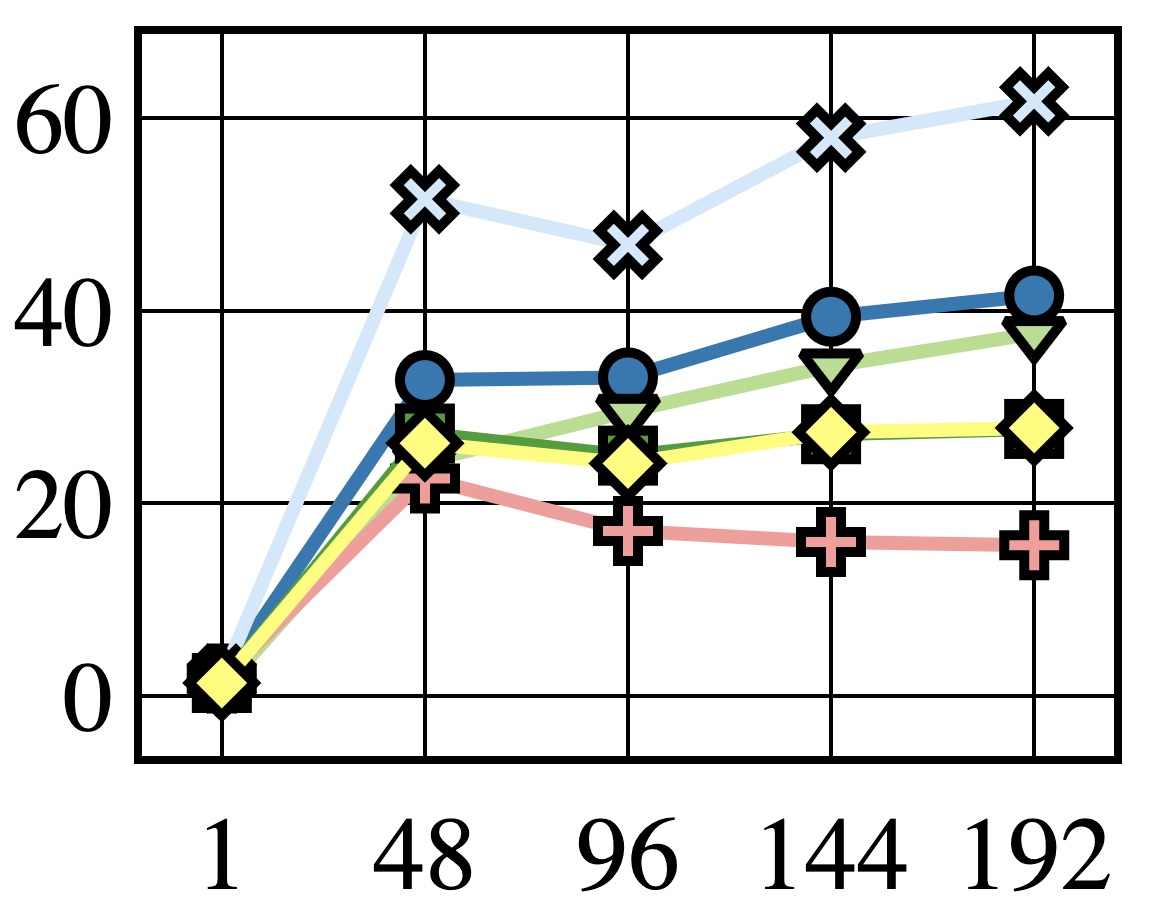}
        \caption{SL, $50-48-2$}
    \end{subfigure}
    \begin{subfigure}{0.17\textwidth}
        \includegraphics[width=\textwidth]{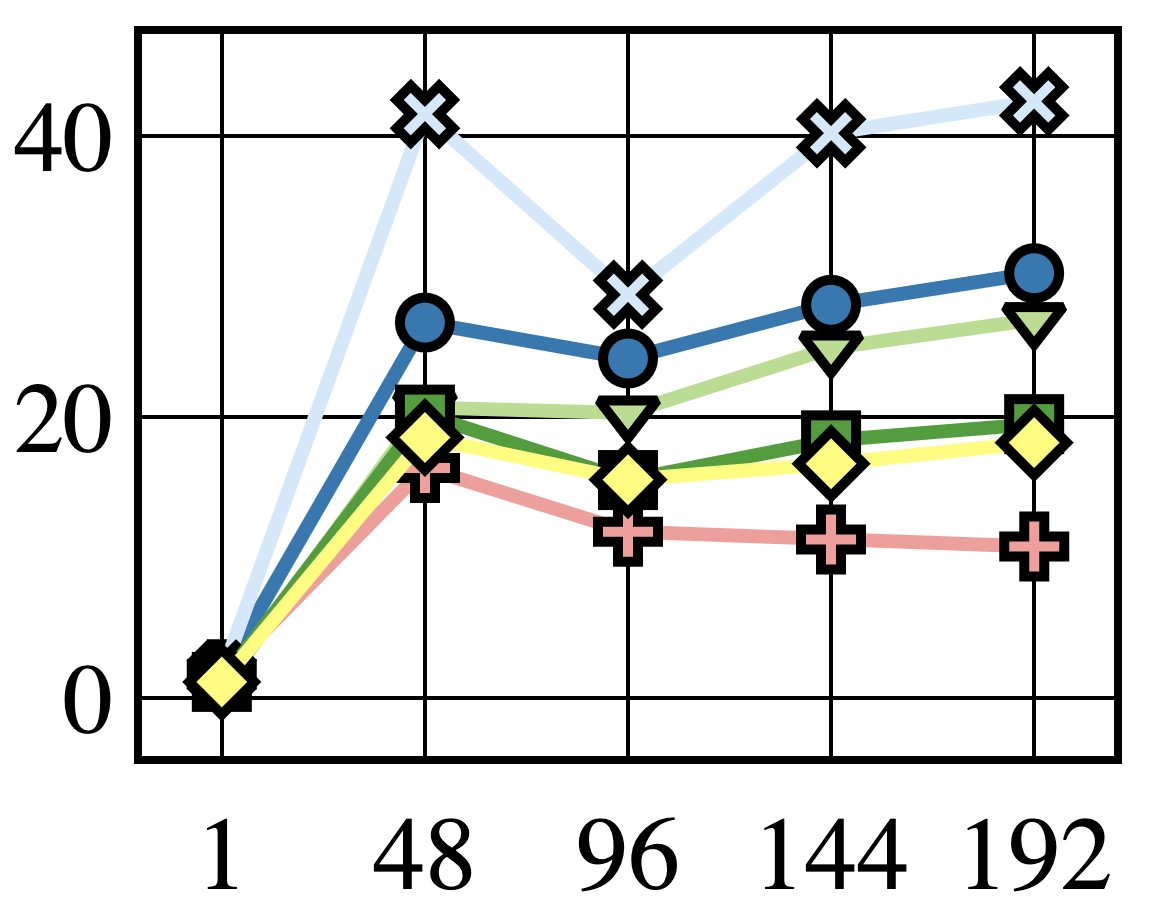}
        \caption{SL, $90-8-2$}
    \end{subfigure}
    \begin{subfigure}{0.17\textwidth}
        \includegraphics[width=\textwidth]{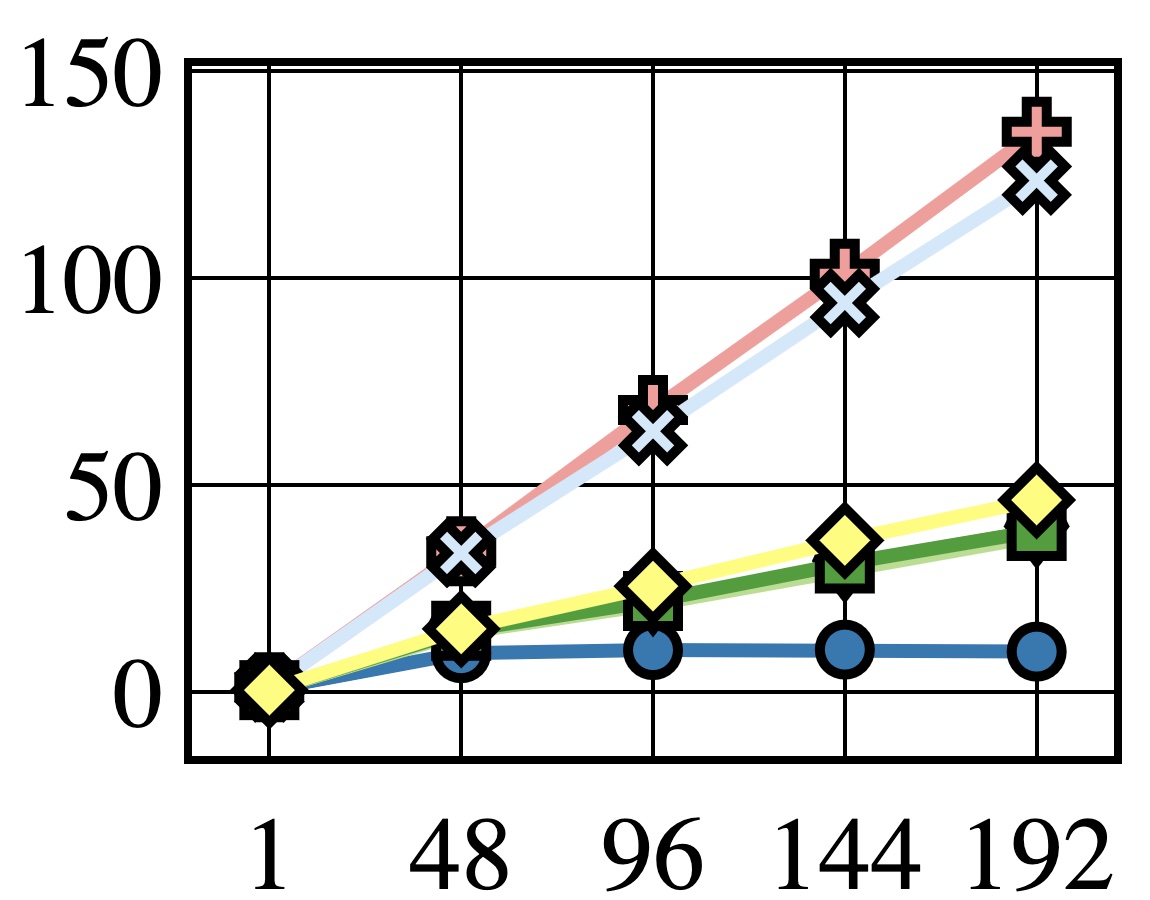}
        \caption{SL, $0-50-50$}
    \end{subfigure}
    \begin{subfigure}{0.17\textwidth}
        \includegraphics[width=\textwidth]{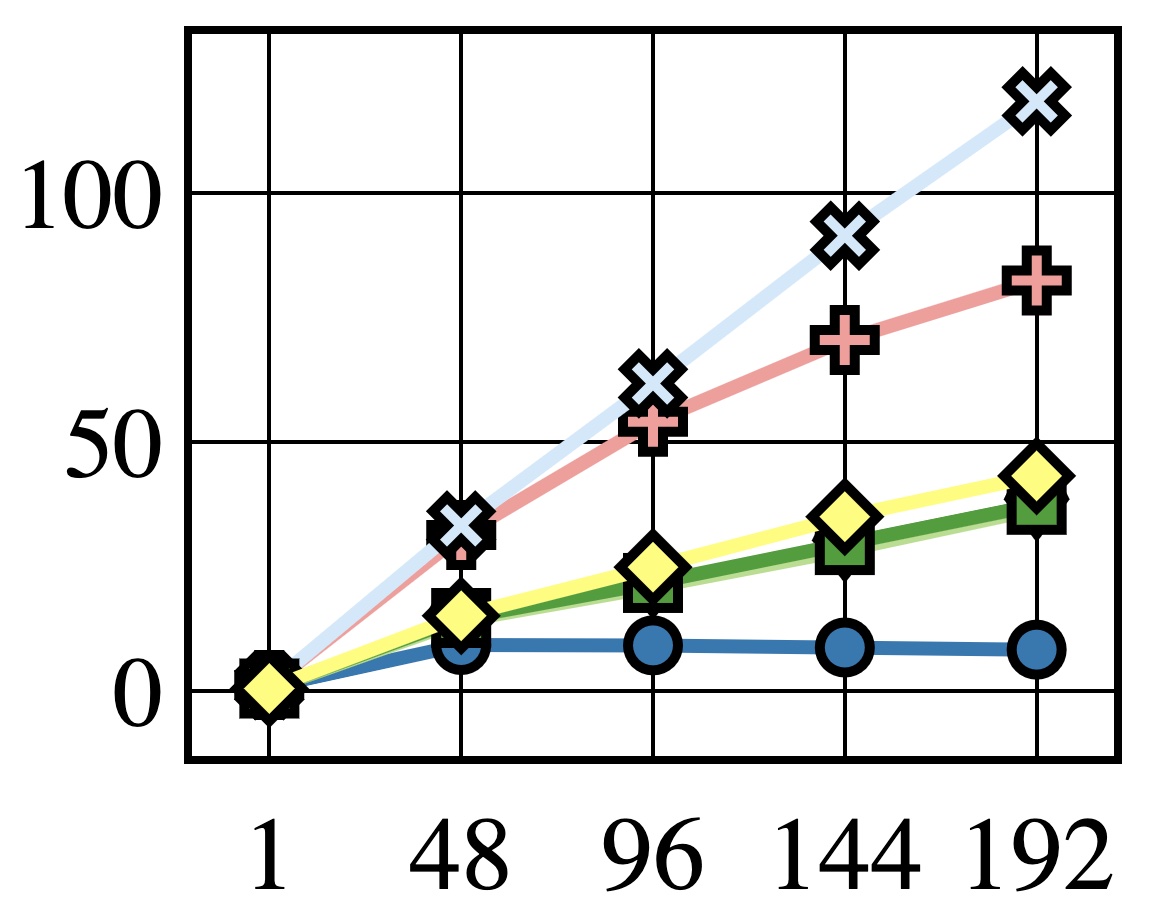}
        \caption{CT, $2-48-50$}
    \end{subfigure}
    \begin{subfigure}{0.17\textwidth}
        \includegraphics[width=\textwidth]{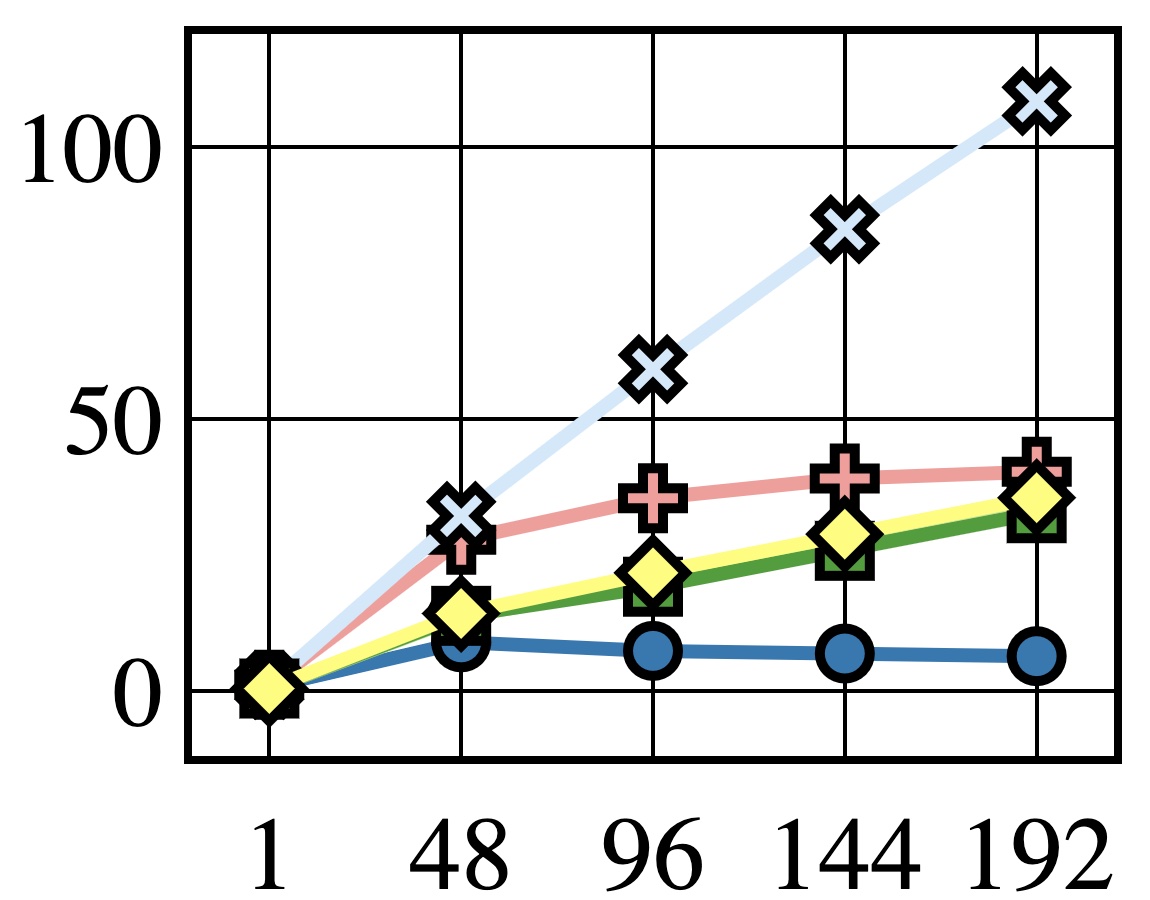}
        \caption{CT, $10-40-50$}
    \end{subfigure}
    \begin{subfigure}{0.17\textwidth}
        \includegraphics[width=\textwidth]{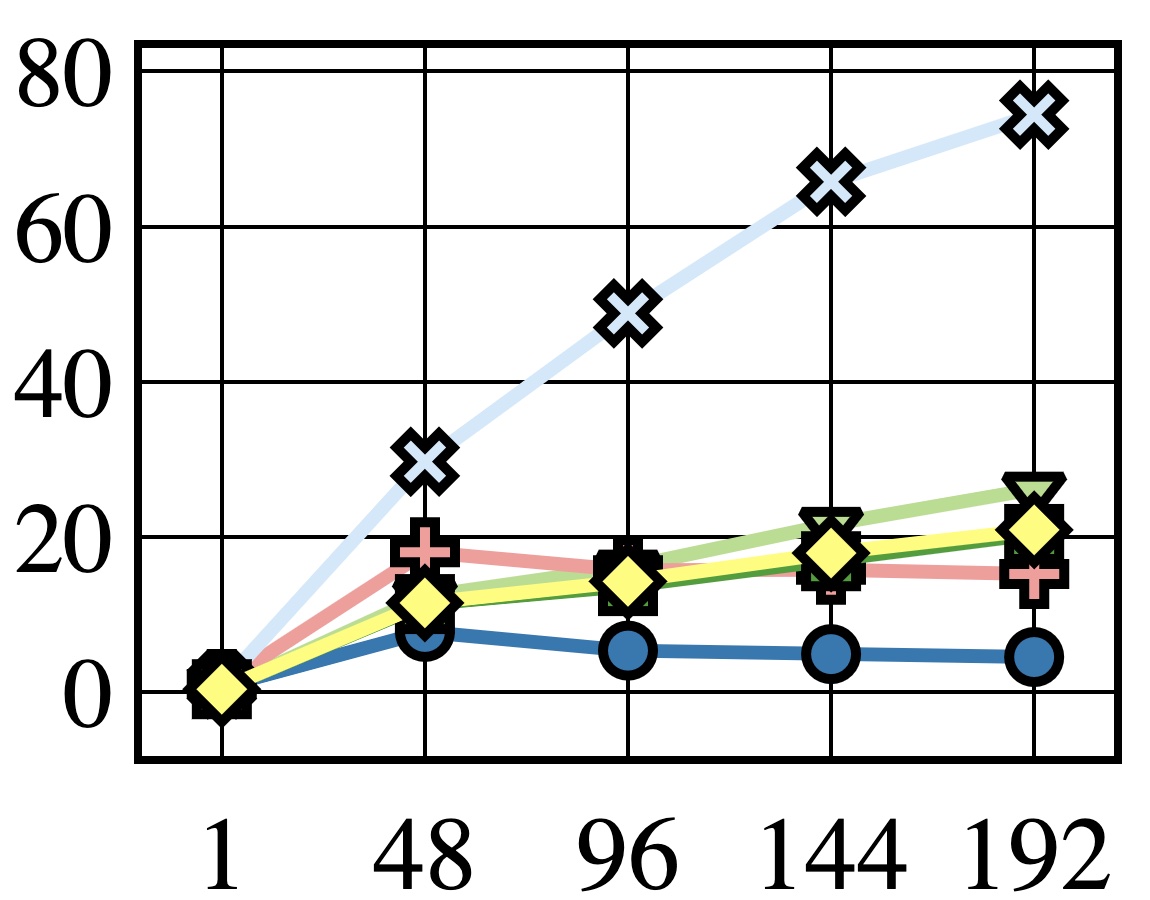}
        \caption{CT, $50-0-50$}
    \end{subfigure}
    \caption{Throughput (Mops/s) under various workload configurations for the skip list (SL) and Citrus tree (CT), with the number of threads on the x-axis. Workloads are written as $U-C-RQ$, corresponding to the percentages of update ($U$), contains ($C$) and range queries ($RQ$).}
    \label{fig:extra-plots}
\end{figure}

\end{document}